\documentclass[sigconf]{acmart}

\renewcommand\footnotetextcopyrightpermission[1]{} 
\usepackage{booktabs} 
\usepackage{amsmath}
\usepackage[boxed,ruled,vlined,linesnumbered]{algorithm2e}

\usepackage{multicol}
\usepackage{subfig}

\setcopyright{none}

\acmConference[]{}
\acmYear{2019}
\copyrightyear{2019}

\newcommand{\eg}{{\em e.g.}}
\newcommand{\ie}{{\em i.e.}}

\DeclareMathOperator*{\argmax}{arg\,max}

\begin{document}
\title[Synthesis of Safe Digital Controllers for Sampled-Data Stochastic  Systems]{Automated  Synthesis of Safe Digital Controllers for \\ Sampled-Data Stochastic Nonlinear Systems}

\author{Fedor Shmarov}
\affiliation{%
  \institution{School of Computing, Newcastle University, UK}
}

\author{Sadegh Soudjani}
\affiliation{%
  \institution{School of Computing, Newcastle University, UK}
}

\author{Nicola Paoletti}
\affiliation{%
  \institution{Royal Holloway, University of London, UK}
}

\author{Ezio Bartocci}
\affiliation{%
  \institution{Faculty of Informatics, TU Wien, Austria}
}

\author{Shan Lin}
\affiliation{%
  \institution{Department of Electrical and Computer Engineering, Stony Brook University, USA}
}

\author{Scott A. Smolka}
\affiliation{%
  \institution{Department of Computer Science, Stony Brook University, USA}
}

\author{Paolo Zuliani}
\affiliation{%
  \institution{School of Computing, Newcastle University}
  \city{Newcastle upon Tyne}
  \country{UK}
}

\renewcommand{\shortauthors}{F. Shmarov et al.}

\begin{abstract}
We present a new method for the automated synthesis of digital controllers with formal safety guarantees for systems with nonlinear dynamics, noisy output measurements, and stochastic disturbances. Our method derives digital controllers such that the corresponding closed-loop system, modeled as a \textit{sampled-data stochastic control system}, satisfies a safety specification with probability above a given threshold. The proposed synthesis method alternates between two steps: generation of a candidate controller $\mathbf{p}_c$, and verification of the candidate. $\mathbf{p}_c$ is found by maximizing a Monte Carlo estimate of the safety probability, and by using a non-validated ODE solver for simulating the system. Such a candidate is therefore sub-optimal but can be generated very rapidly. To rule out unstable candidate controllers, we prove and utilize Lyapunov's indirect method for instability of sampled-data nonlinear systems. In the subsequent verification step, we use a validated solver based on SMT (Satisfiability Modulo Theories) to compute a numerically and statistically valid confidence interval for the safety probability of $\mathbf{p}_c$. If the  probability so obtained is not above the threshold, we expand the search space for candidates by increasing the controller degree.
We evaluate our technique on three case studies: an artificial pancreas model, a powertrain control model, and a quadruple-tank process.
\end{abstract}

\maketitle

\section{Introduction}
\label{sec:intro}



Digital control \cite{Ogata95} is 
essential in many cyber-physical and embedded systems applications, ranging from aircraft autopilots
to biomedical devices, due to its superior flexibility and scalability, and lower cost compared to its
analog counterpart. The synthesis of analog controllers for linear systems is well-studied \cite{Nise16}, but
its extension to nonlinear and stochastic systems has proven much more challenging. Furthermore, digital control adds
extra layers of complexity, \eg, time discretization and signal quantization. A common problem
in both digital and analog control is the lack of automated synthesis techniques with provable guarantees,
especially for properties beyond stability (\eg, safety) for nonlinear stochastic systems. 

In this paper we address
this problem by introducing a new method for the {\em synthesis of probabilistically safe digital controllers} for a large class of stochastic nonlinear systems, viz.\ {\em sampled-data stochastic control systems}. In such systems, the plant is
a set of nonlinear differential equations subject to random disturbances, and
the digital controller samples the noisy plant output, generating the control input with a fixed frequency.

Controllers are usually designed to achieve stability of the closed-loop system. 
\emph{Lyapunov's indirect method} provides conditions under which the stability of an equilibrium point of a
nonlinear system follows from the stability of that point for the linearized version of the system \cite{khalil2002}.
Lyapunov's method for sampled-data nonlinear systems is much more involved. Previous work~\cite{Teel04} provides
{\em sufficient} conditions on the sampled-data linearized system that ensure stability of the sampled-data nonlinear system.
Unfortunately, it is difficult to verify these conditions algorithmically. In this paper, we instead prove
{\em necessary} conditions for stability that are easy to verify, and use them to restrict the controller synthesis domain.
However, a stable system is not necessarily safe, as during the transient the system might reach an unsafe, catastrophic state. The synthesis approach that we propose overcomes this issue by deriving controllers that are safe.

Given an invariant $\phi$ (\ie, a correctness specification), and a nonlinear plant with stochastic disturbances and noisy outputs, our method synthesizes a digital controller such that the corresponding closed-loop system satisfies $\phi$ with probability above a given threshold $\vartheta$. The synthesis algorithm (Algorithm~\ref{alg:main} in Section~\ref{sec:method}) is illustrated in Figure~\ref{fig:approach}.  It works by alternating between two steps: generation of a candidate controller $\mathbf{p}_c$, and verification of the candidate. $\mathbf{p}_c$ is generated via the \textbf{optimize} procedure (see Algorithm~\ref{algo:optimize_func}), which maximizes a Monte Carlo estimate of the satisfaction probability by simulating a discrete-time approximation of the system with a non-validated ODE solver. Such a candidate is, therefore, sub-optimal but very rapid to generate.  To rule out unstable controller candidates, we prove and utilize Lyapunov's indirect method for instability of sampled-data nonlinear systems. Along with $\mathbf{p}_c$, \textbf{optimize} returns an approximate confidence interval (CI) $[a,b]$ for the satisfaction probability.

Next, in the verification step (procedure \textbf{verify}), we use a validated solver based on SMT (Satisfiability Modulo Theories) to compute a numerically and statistically valid CI $[a',b']$ for the satisfaction probability of $\mathbf{p}_c$. If the deviation between the approximate CI $[a,b]$ and the precise CI $[a',b']$ is too large, indicating that the candidates generated by \textbf{optimize} are not sufficiently accurate, we increase the precision of the non-validated, fast solver (procedure \textbf{update\_discretization}). If instead the precise probability is not above the threshold $\vartheta$, we expand the search space for candidates by increasing the controller degree. 

Summarizing, the novel contributions of this paper are:
\begin{itemize}
	\item we synthesize digital controllers for nonlinear systems subject to stochastic disturbances and measurement noise, while state-of-the-art
		approaches consider linear systems only;
	\item we prove Lyapunov's indirect method for instability of nonlinear systems in closed-loop
		with digital controllers;
	\item we present a novel algorithm that synthesizes digital controllers with guaranteed probabilistic safety
		properties.
\end{itemize}

\begin{figure*}
\centering
\includegraphics[width=.95\textwidth, trim=0 .5cm 11cm 10cm, clip]{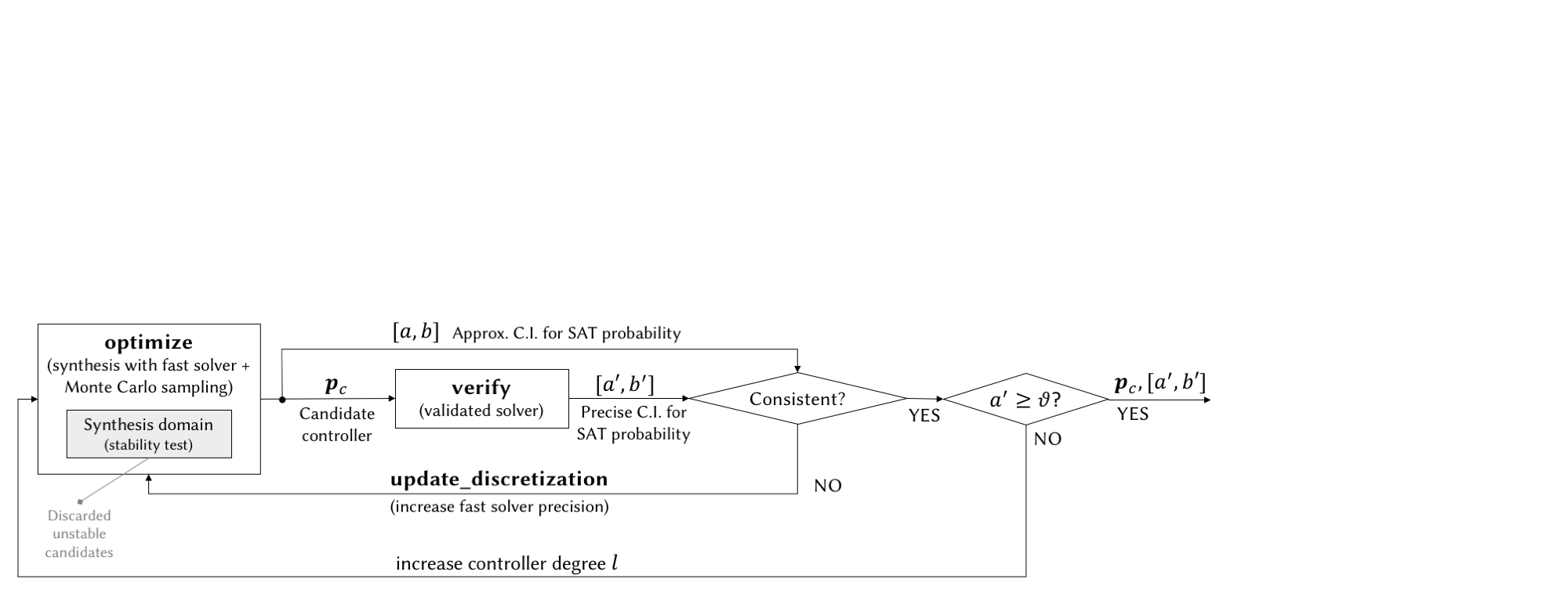}
\caption{Overview of the synthesis approach.}
\label{fig:approach}
\end{figure*}

\section{Sampled-data Stochastic Systems}
\label{sec:sdss}
We consider \textit{sampled-data stochastic control systems (SDSS)}, a rich class of control systems where the plant is specified as a nonlinear system subject to
random disturbances. The controller periodically samples the plant output subject to random noise generating, using the plant output's history, a control input that is kept constant during the sampling period with a zero-order hold; see Figure~\ref{fig:sdss}. The controller is characterized by a number of unknown parameters, which are the
target of our synthesis algorithm. 

\begin{definition}[Sampled-data Stochastic Control System]\label{defn:sdss}
An SDSS can be described in the following state-space notation:
\begin{align}
\frac{d}{dt}x(t) = & \ f(x(t), u(t), \mathbf{d}(t)), \ x(0) = x_0\nonumber\\
y(t_k) = & \ o(x(t_k)) + \eta(t_k), \ t_k=k\cdot \tau, \ k \in \mathbb{Z}^{\geq 0} \nonumber \\
	u(t) = & \ h(y(t_0),\ldots,y(t_k),u(t_0),\ldots,u(t_k),\mathbf{p}), \ \forall t\mathord{\in} [t_{k},t_{k+1})
\label{eq:input}
\end{align}
where $x(\cdot)\in \mathbb{R}^n$ is the state of the plant; $x_0$ is the initial state at time $t=0$; $\mathbf{d}(\cdot) \in \mathbb{R}^q$ is the disturbance; $y(\cdot)$ is the plant output, which is a function of the state with additive i.i.d. noise $\eta \sim \mathcal{N}(0,W)$ with covariance matrix $W$; $u(\cdot) \in \mathbb{R}^m$ is the control input, updated at every sampling period $\tau > 0$ by the digital controller $h$ (defined in Section~\ref{sec:digi_ctrl}); and $\mathbf{p} \in \mathbb{P} \subset \mathbb{R}^{2L+1}$ is the vector of unknown controller parameters, where $\mathbb{P}$ is a hyperbox (\ie, a product of closed intervals).
The dynamics of the plant is governed by the vector field $f$, which is assumed to be in $C^1$, hence Lipschitz-continuous. We also assume that the output map $o(\cdot)$ is in $C^1$.
\end{definition}

We assume that there are no time lags for transmitting the plant output to the controller and the control input to the plant.
The disturbance $\mathbf{d}(\cdot)$ is a piecewise-continuous function having, for a time horizon $T$, a finite number of discontinuities. The discontinuity points and the value of $\mathbf{d}(\cdot)$ at each sub-domain can be defined in terms of a finite number of random parameters drawn from arbitrary distributions. Note that these assumptions on $\mathbf{d}(\cdot)$ are reasonably mild and allow us to define very general classes of systems, which subsume, for instance, numerical solutions of stochastic differential equations~\cite{ruemelin1982numerical}.\footnote{Such numerical solutions rely on computing the value of the Wiener process at discrete time points, which makes it a special case of our disturbances.}

\begin{figure}
\centering
\includegraphics[width=\columnwidth, trim=1cm 13.5cm 6cm 0, clip]{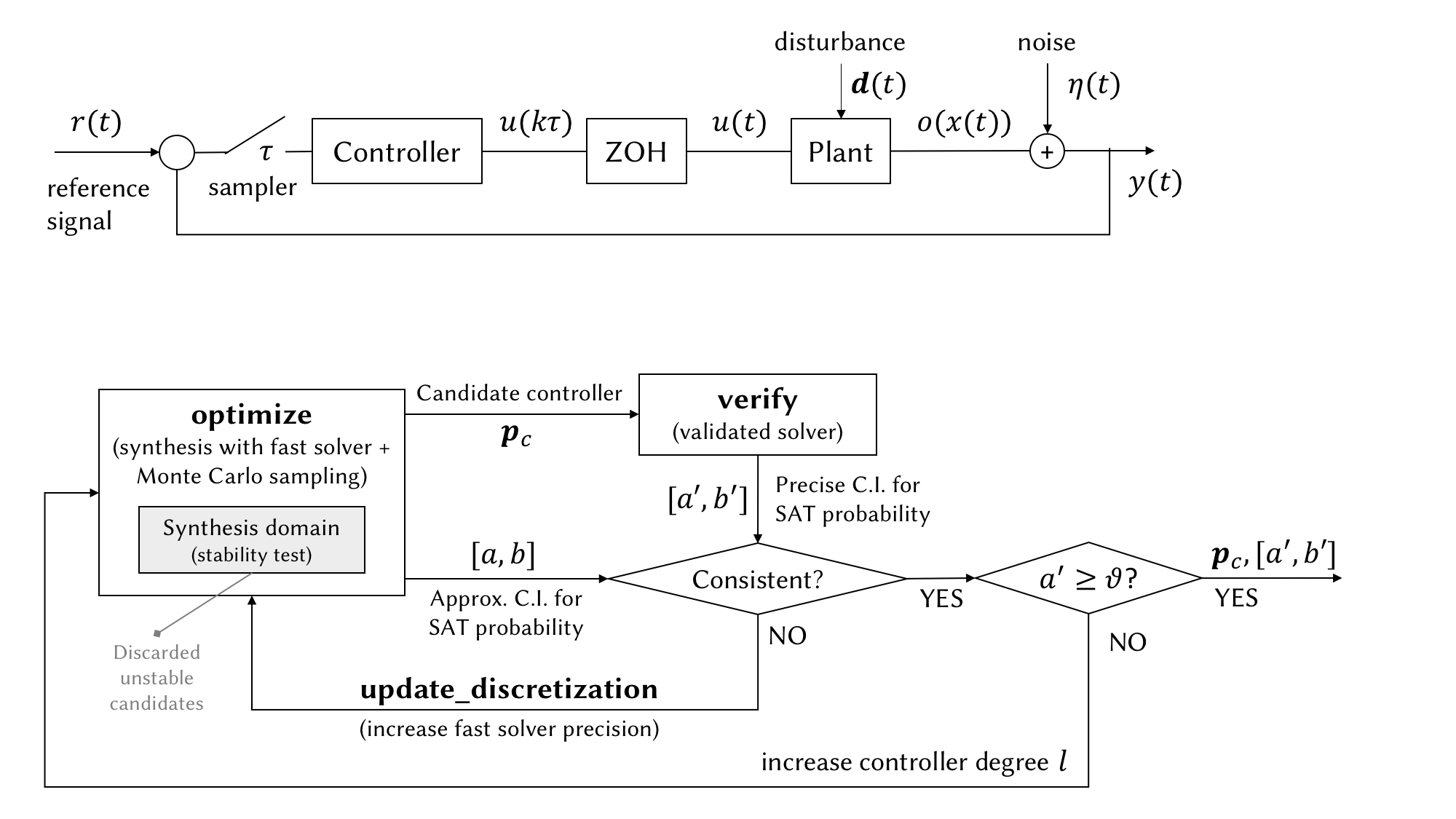}
\caption{Diagram of a sampled-data stochastic control system.}
\label{fig:sdss}
\end{figure}

\section{Digital Controllers}\label{sec:digi_ctrl}


The operation of a digital controller is succinctly indicated in Equation~\eqref{eq:input}. These computations are generally performed using current and past output samples and past input samples.

\begin{definition}[Digital Controller for SDSS]\label{defn:digital_controller}
Given an SDSS, we denote $y(k) = y(t_k)$ and $u(k) = u(t_k)$ and define the tracking error as
\begin{equation}
\label{eq:feedback}
e(k) = r(k)-y(k),\quad k\in\mathbb Z^{\ge 0}
\end{equation}
where $r(\cdot)$ is the reference signal. The output of the controller is
\begin{equation}
u(k) = -\sum_{i=1}^{L}a_i u(k-i) + \sum_{i=0}^{L} b_i e(k-i),
\label{eq:digital_time}
\end{equation}
where $u(j)=e(j)=0$ for $j<0$\footnote{Note that if the controller has been previously deployed, i.e., it starts from a non-empty history, then $u(j),e(j)$ may be nonzero for $j<0$.}
and $L$ is the controller degree.
\end{definition}
Controller design amounts to finding a degree $L$ and coefficients $\{a_i\}_{i=1}^L, \{b_i\}_{i=0}^L$ that ensure the desired behavior of the closed-loop system. The vector of parameters $\mathbf{p}$ defined in (\ref{eq:input}) is recovered by setting $\mathbf{p} = [b_0,a_1, b_1,\ldots,a_L,b_L]^T$.
An alternative description of the controller is via the state-space representation
\begin{align}
x^c(k+1) & = G_c x^c(k)+ H_c e(k),\quad x^c(0) = x^c_0\nonumber\\
u(k) & = C_c x^c(k) + D_c e(k),\quad k\in\mathbb Z^{\ge 0}, \label{eq:digital_state}
\end{align}
where $x^c(k)$ is the state of the controller and matrices $(G_c,H_c,C_c,D_c)$ need to be designed. The above two representations are equivalent. Given a controller in the form of \eqref{eq:digital_time}, one can transform it to the representation \eqref{eq:digital_state}, for instance by taking states as memories that store previous values of inputs/outputs. Given matrices $(G_c,H_c,C_c,D_c)$, one can easily compute coefficients $\{a_i,b_i\}$ in \eqref{eq:digital_time} using matrix multiplications \cite{Ogata95}.




\subsection{Stability of the closed-loop system}
A necessary requirement of any controller is stability of the closed-loop system. In our setting, we have the influence of both external inputs $(\mathbf{d},\eta,r)$ and initial states $(x_0,x_0^c)$. A suitable notion is \emph{input-to-state stability (ISS)} \cite{sontag2008input}, which implies that bounded input signals must result in bounded outputs and, at the same time, that  the effect of initial states must disappear as time goes to infinity. A necessary requirement of ISS is Lyapunov stability of the system `without' external inputs, stated in the next definition, a requirement that can be applied to both continuous- and discrete-time systems.

\begin{definition}
Consider a dynamical system with state space $\mathcal D\subset \mathbb R^n$ and without any external inputs,
where $x_e\mathord{\in} \mathcal D$ is an equilibrium point and $\mathcal D$ is open. Then $x_e$ is called
\textbf{Lyapunov stable} if for every $\epsilon>0$ there exists a $\delta>0$ such that for all $x(0)\in\mathcal D$ with $\|x(0)-x_e\|\le \delta$, we have $\|x(t)-x_0\|\le \epsilon$ for all $t\ge 0$.
\end{definition}


It is very easy to verify stability for discrete-time linear systems.
\begin{proposition}[Stability of Linear Systems \cite{khalil2002}]\label{prop:linear_stab}
A linear discrete-time system $x(k+1) = G x(k)$ is Lyapunov stable at $x_e = 0$ and $\lim_{k\rightarrow\infty} x(k) = 0$ if and only if all eigenvalues of $G$ are inside unit circle. This condition is equivalent to the existence of positive definite matrices $M,Q$ such that $G^T M G - M = -Q$.
\end{proposition}

In the remainder of this section, we consider a version of SDSS in Definition~\ref{defn:sdss} controlled by \eqref{eq:digital_state} without any external input, \ie, when $(\mathbf d,\eta,r)$ are identically zero. We study Lyapunov stability of the closed-loop system without external inputs, which is necessary for having input-to-state stability.
Let us put $(\mathbf d,\eta,r)\equiv 0$ and define $x_1 := x-x_e$, $u_1 := u-u_e$, with $(x_e,u_e)$ being the equilibrium point of SDSS \eqref{eq:input}, \ie, $f(x_e,u_e,0) = 0$. Similarly, define $y_1 := y-y_e$ with $y_e = o(x_e)$ and $x_1^c  := x_c - x_e^c$  with $x_e^c$ being the equilibrium point for the controller.
%
%
%
%
We then denote the plant dynamics after eliminating external inputs based on shifted version of variables by
\begin{align}
\nonumber
\frac{d}{dt}x_1(t) & = \bar f(x_1(t),u_1(t)),\\
y_1(t_k) & = \bar o(x_1(t_k)),\quad t_k = k\tau,\,\, k\in\mathbb Z^{\ge 0}
\label{eq:nonl_additive}
\end{align}
where $\bar f(x_1,u_1) = f(x_1+x_e,u_1+u_e,0)$ and $\bar o(x_1) = o(x_1+x_e)-y_e.$
Thus $x_1 = 0, u_1 = 0$ is an equilibrium point for \eqref{eq:nonl_additive}.
%
%
%
The digital controller dynamics is likewise given by
\begin{align}
x^c_1(k+1) & = G_c x^c_1(k) - H_c y_1(k\tau)\nonumber\\
u_1(k\tau) & = C_c x^c_1(k) - D_c y_1(k\tau),\label{eq:digital}
\end{align}
where the minus sign is due to $r(\cdot)=0$ and negative feedback in \eqref{eq:feedback}.
The nonlinear system \eqref{eq:nonl_additive} is controlled with the digital controller \eqref{eq:digital} by setting $u_1(t) = u_1(k\tau)$ for all $t\in[k\tau,(k+1)\tau)$, $k\in\mathbb Z^{\ge 0}$.

Ensuring stability of the sampled-data nonlinear control system \eqref{eq:nonl_additive}-\eqref{eq:digital} is difficult in general. Sufficient conditions for preserving stability under linearization are provided in \cite{Teel99,Teel04}, but they are hard to verify automatically. Rather, we provide an easy-to-check necessary condition for Lyapunov stability to reject unsuitable controllers. This necessary condition is based on Lyapunov's indirect method developed here for sampled-data nonlinear systems.
In the following we prove that if the linearized closed-loop system has an eigenvalue outside the unit circle, the nonlinear closed-loop system \eqref{eq:nonl_additive}-\eqref{eq:digital} is not Lyapunov stable, thus the system \eqref{eq:input}-\eqref{eq:digital_state} is not input-to-state stable.

We first consider the linearized version of the closed-loop system \eqref{eq:nonl_additive}-\eqref{eq:digital},
which is
\begin{align*}
\frac{d}{dt}x_2(t) & = A x_2(t) + B u_2(k\tau),\quad\forall t\in[k\tau,(k+1)\tau)\nonumber\\
x^c_2(k+1) & = G_c x^c_2(k)- H_c C x_2(k\tau)\nonumber\\
u_2(k\tau) & = C_c x^c_2(k) - D_c C x_2(k\tau),
\end{align*}
where
$A := \frac{\partial \bar f}{\partial x_1}(0,0)$,
$B := \frac{\partial \bar f}{\partial u_1}(0,0)$,
and
$C := \frac{\partial \bar o}{\partial x_1}(0)$,
Define $g(x_1,u_1) := \bar f(x_1,u_1) - A x_1 - B u_1$ and $l(x_1) := \bar o(x_1)-C x_1$ \ie, the non-linear terms describing the deviation between non-linear and linearized functions. Thus, we have
\begin{equation*}
\lim\limits_{\|(x_1,u_1)\|\rightarrow 0} \frac{\|g(x_1,u_1)\|}{\|(x_1,u_1)\|}= 0
\quad\text{and}\quad
\lim\limits_{\|x\|\rightarrow 0} \frac{\|l(x_1)\|}{\|x_1\|}= 0.
\end{equation*}
Then, for any $\gamma>0$ there exists an $r(\gamma)>0$ such that
\begin{equation}
\label{eq:gamma}
\|g(x_1,u_1)\|\le \gamma \|x_1\| + \gamma \|u_1\|
\quad\text{and}\quad
\|l(x_1)\|\le \gamma \|x_1\|,
\end{equation}
for all $x_1\in\mathbb R^n, u_1\in\mathbb R^m$ with
$\|(x_1,u_1)\|\le r(\gamma)$.
We now simplify the dynamics of the closed-loop nonlinear system as
\begin{equation}
\label{eq:simplified}
\begin{cases}
\frac{d}{dt}x_1(t) = A x_1(t) + B C_c x^c_1(k) - B D_c C x_1(k\tau) + g(x_1(t),u_1(k\tau))\\
x^c_1(k+1) = G_c x^c_1(k) - H_c C x_1(k\tau) - H_c l(x_1(k\tau)),
\end{cases}
\end{equation}

The next lemma establishes a bound on $x_1(t)$ for any  $t\in [k\tau,k\tau+\tau]$, as a function of $x_1(k\tau)$ and $x^c_1(k)$. Due to space constraints we present the
proof of this lemma in the appendix.
%
\begin{lemma}
\label{lem1}
Under dynamics \eqref{eq:simplified}, for a given $t\in[k\tau,(k+1)\tau]$ and any $\gamma>0$, we have
\begin{equation}
\label{eq:bound_pert}
\|x_1(t)\|\le h_1(t-k\tau,\gamma)\|x_1(k\tau)\| + h_2(t-k\tau,\gamma)\|x^c_1(k)\|
\end{equation}
if $\|(x_1(t_1),u_1(k\tau))\|\le r(\gamma)$ for all $t_1\in[k\tau,t]$, with $r(\gamma)$ satisfying property \eqref{eq:gamma}. Functions $h_1,h_2$ are continuous and nonnegative with $h_1(0,\gamma) = 1$ and $h_2(0,\gamma) = 0$.
\end{lemma}
The upper bound \eqref{eq:bound_pert} enables us to study the effect of the nonlinear terms $g(\cdot)$ and $l(\cdot)$ in the sampled version of the dynamics, which can be written as
\begin{align*}
x_1((k+1)\tau) & = (G-H D_c C) x_1(k\tau) - H C_c x^c_1(k) + \hat g(k\tau),\nonumber\\
x^c_1(k+1) & = G_c x^c_1(k) - H_c C x_1(k\tau) + \hat l(k\tau),
\end{align*}
with $G = e^{A\tau}$, $H = \int_0^\tau e^{A\lambda}Bd\lambda$, $\hat l(k\tau) = - H_c l(x_1(k\tau))$, and
\begin{equation}
\label{eq:hat_g}
\hat g(k\tau) = \int_{0}^{\tau}e^{A(\tau-\lambda)}g(x_1(k\tau+\lambda),u_1(k\tau))d\lambda.
\end{equation}
Next, we derive a bound for $\hat{g}(\cdot)$ in terms of $x_1$ and $x_1^c$.
\begin{lemma}
\label{lem2}
For any $\gamma>0$, there exist continuous functions $\hat h_1,\hat h_2$ such that the following inequality holds for $\hat g(\cdot)$ defined in \eqref{eq:hat_g},
\begin{equation*}
\|\hat g(k\tau)\|\le \gamma \hat h_1(\tau)\|x_1(k\tau)\| + \gamma \hat h_2(\tau)\|x^c_1(k)\|,
\end{equation*}
if $\|(x_1(t),u_1(k\tau))\|\le r(\gamma)$ for all $t\in[k\tau,(k+1)\tau]$, with $r(\gamma)$ satisfying property \eqref{eq:gamma}.
Functions $\hat h_1,\hat h_2$ are nonnegative with $\hat h_1(0) = 0$ and $\hat h_2(0) = 0$.
\end{lemma}
The explicit form of $\hat h_1,\hat h_2$ is provided, along with the proof, in the appendix.
We are now ready to state our main result of this section.
\begin{theorem}
\label{thm:instability}
The continuous-time nonlinear system \eqref{eq:nonl_additive} controlled with the digital controller \eqref{eq:digital} is unstable
if the linearized continuous-time system controlled by the same digital controller has a pole outside the unit circle.
\end{theorem}

\textbf{Sketch of the proof.} We prove the theorem by contradiction. We show that there is an $r>0$ such that for any $\delta>0$ we can find an initial state for the system and the controller with $\|(x_1(0),x_1^c(0))\|\le \delta$ and time $t>0$ with $\|(x_1(t),x_1^c(t))\|>r$. We construct a Lyapunov function for the linear system that is strictly increasing on a suitable set of initial states. By a proper selection of $r$, we show that this Lyapunov function is also strictly increasing on the nonlinear system if the trajectory remains inside the ball with radius $r$, which is a contradiction.
\begin{proof}
The closed-loop linearized system will have the following dynamics in discrete time
\begin{align*}
x_2((k+1)\tau) & = (G-H D_c C) x_2(k\tau) - H C_c x^c_2(k),\nonumber\\
x^c_2(k+1) & = G_c x^c_2(k) - H_c C x_2(k\tau),
\end{align*}
with $G = e^{A\tau}$ and $H = \int_0^\tau e^{A\lambda}Bd\lambda$.
This gives the following state transition matrix
\begin{equation}
\label{eq:hat_G}
\hat G := \left[\begin{array}{cc}
G - H D_c C & - H C_c\\
G_c & - H_c C
\end{array}
\right],
\end{equation}
which is assumed to have at least one eigenvalue outside the unit circle.
We cluster the eigenvalues of $\hat G$ into a group of eigenvalues outside the unit circle and a group of eigenvalues on or inside the unit circle. Then there is a nonsingular matrix $T$ such that
\begin{equation*}
T \hat G T^{-1} = \left[
\begin{array}{cc}
G_1 & 0\\
0 & G_2
\end{array}
\right],
\end{equation*}
where $G_1^{-1}$ is stable. In other words, $G_1$ contains all of the eigenvalues of $\hat G$ from the first group.
(The matrix $T$ can be found for instance by transforming $\hat G$ into its real Jordan form.)
Let us define
\begin{equation*}
y(k) :=
\left[
\begin{array}{c}
x_1(k\tau)\\
x_1^c(k)
\end{array}
\right],
\quad
z = Ty = \left[
\begin{array}{c}
z_1\\
z_2
\end{array}
\right] =
\left[
\begin{array}{c}
T_1 y\\
T_2 y
\end{array}
\right],
\end{equation*}
where the partitions of $z$ and $T$ are compatible with the dimensions of $G_1$ and $G_2$. The dynamics of $z$ becomes
\begin{equation*}
z(k+1) = T\hat G T^{-1} z(k) + T
\left[
\begin{array}{c}
\hat g(k\tau)\\
\hat l(k\tau)
\end{array}
\right]
 = \left[
\begin{array}{c}
G_1 z_1(k) + \hat g_1(k\tau)\\
G_2 z_2(k) + \hat g_2(k\tau)
\end{array}
\right].
\end{equation*}
Now define $\delta>0$ by $$1+2\delta = \min_i|\lambda_{i}(G_1)|.$$
Then both $G_2/(1+\delta)$ and $(1+\delta)G_1^{-1}$ are stable matrices.
According to Proposition~\ref{prop:linear_stab}, there are positive definite matrices $M_1,M_2$ and $Q_1,Q_2$ such that the following matrix equalities hold
\begin{equation*}
(1+\delta)^2{G_1^T}^{-1} M_1 G_1^{-1} - M_1 = -Q_1,\quad
G_2^T M_2 G_2/(1+\delta)^2 - M_2 = -Q_2.
\end{equation*}
and we get
\begin{align*}
& G_1^T M_1 G_1 - M_1 = \bar Q_1 + (\delta^2+2\delta)M_1,\\
& G_2^T M_2 G_2 - M_2 = -\bar Q_2 + (\delta^2+2\delta) M_2,
\end{align*}
with $\bar Q_1 := G_1^T Q_1 G_1$ and $\bar Q_2:= (1+\delta)^2 Q_2$.
For the function
$$V(k) := z_1(k)^T M_1 z_1(k) - z_2(k)^T M_2 z_2(k),$$
we have
\begin{align*}
V&(k+1) - V(k) = z_1(k+1)^T M_1 z_1(k+1) - z_1(k)^T M_1 z_1(k)\\
& - z_2(k+1)^T M_2 z_2(k+1) + z_2(k)^T M_2 z_2(k)\\
 & = z_1(k)^T (G_1^T M_1 G_1 - M_1)z_1(k)
 + 2\hat g_1(k\tau)^T M_1(G_1 z_1(k)+\hat g_1(k\tau))\\
 &  - z_2(k)^T (G_2^T M_2 G_2 - M_2)z_2(k)
 - 2\hat g_2(k\tau)^TM_2(G_2 z_2(k)+\hat g_2(k\tau))\\
 & \ge  z_1^T \bar Q_1 z_1 + (\delta^2+2\delta)z_1^T M_1 z_1 + z_2^T \bar Q_2 z_2 - (\delta^2+2\delta)z_2^T M_2 z_2\\
 & - \gamma c_1 \|z\|^2 - \gamma^2 c_2 \|z\|^2,
 \end{align*}
 where the positive constant
 \begin{equation*}
 c_2 := 2\lambda_{max}(M_2)\|T\|^2 \|T^{-1}\|^2\left(\hat h_1(\tau)^2 + \hat h_2(\tau)^2 + 1 \right)
 \end{equation*}
 is obtained using definition of $\hat g_1,\hat g_2$ as a function of $\hat g$ and Lemma \ref{lem2} as
 \begin{align*}
 & 2\hat g_1(k\tau)^T M_1 \hat g_1(k\tau) - 2 \hat g_2(k\tau)^TM_2 \hat g_2(k\tau)
 \ge -2\lambda_{max}(M_2)\|\hat g_2(k\tau)\|^2\\
 & \ge -2\lambda_{max}(M_2)\|T\|^2\left(\|\hat g(k\tau)\|^2 + |\hat l(k\tau)\|^2\right)\\
 & \ge -2\lambda_{max}(M_2)\|T\|^2\gamma ^2\left(\hat h_1(\tau)^2 + \hat h_2(\tau)^2 + 1\right) \|T^{-1}\|^2\|z\|^2\\
 & = -c_2\gamma^2 \|z\|^2.
 \end{align*}
 Similarly, we use the Cauchy-Schwarz inequality for
 \begin{equation*}
 c_1:=2 \left(\|M_1G_1\| + \|M_2G_2\|\right)\|T\|\|T^{-1}\|\sqrt{\hat h_1(\tau)^2 + \hat h_2(\tau)^2 + 1}
 \end{equation*}
 as
 \begin{align*}
 & 2\hat g_1(k\tau)^T M_1G_1 z_1(k)- 2\hat g_2(k\tau)^TM_2G_2 z_2(k)\\
 & \ge
 -2\|\hat g_1(k\tau)\| \|M_1G_1\| \|z_1(k)\|
 -2\|\hat g_2(k\tau)\| \|M_2G_2\| \|z_2(k)\|\\
 & \ge
 -2 \|T\|\left\|[\hat g(k\tau),\hat l(k\tau)]^T\right\| \|z(k)\| (\|M_1G_1\| + \|M_2G_2\|)\\
 & \ge
 -2 \|T\|\gamma\sqrt{\hat h_1(\tau)^2 + \hat h_2(\tau)^2 + 1}\,\, \|T^{-1}\|\|z\|^2(\|M_1G_1\| + \|M_2G_2\|)\\
 & = -c_1 \gamma \|z\|^2.
 \end{align*}
 Then, function $V(\cdot)$ satisfies
 \begin{align*}
 V(k+1)\ge (1+\delta)^2 V(k) + (c_0 - c_1\gamma -c_2\gamma^2)\|z\|^2,
\end{align*}
with $c_0 = \min_{i,j}\{\lambda_i(\bar Q_1),\lambda_j(\bar Q_2)\}$.

Take $0<\gamma_0\le 1$ sufficiently small such that $c_0-c_1\gamma_0-c_2\gamma_0^2\ge 0$ with its associated radius $r_0 = r(\gamma_0)$. Note that this is always possible since $\hat h_1,\hat h_2$, thus $c_1,c_2$, are bounded on the interval $\gamma\in (0,1]$.
Then we have $V(k+1)\ge (1+\delta)^2 V(k)$ as long as $\|(x_1(t),u_1(k\tau))\|\le r_0$. For the proof of instability, take radius
\begin{equation*}
r_1 := \frac{r_0}{1 + \|C_c\| + \|D_c\|+\|D_c C\|}
\end{equation*}
and any initial condition $y(0) = (x_1(0),x_1^c(0))$ such that
\begin{equation*}
V(0) = y(0)^T (T_1^T M_1 T_1 - T_2^T M_2 T_2) y(0)>0.
\end{equation*}
We claim that the trajectory starting from $y(0)$ will always leave the ball with radius $r_1$. Suppose this is not true,
\ie, $\|y(t)\|\le r_1$ for all $t\ge 0$. Then $\|(x_1(t),u_1(k\tau))\|\le r_0$ for all $t\in[k\tau,k\tau+\tau]$ and
$k\in\mathbb Z^{\ge 0}$ and $V(k+1)\ge (1+\delta)^2 V(k)\ge (1+\delta)^{2k}V(0)$. Then $\lim_{k\rightarrow\infty} V(k) = \infty$, which contradicts the boundedness of $y(t)$.
\end{proof}

\paragraph{\bf{Stability test}}

The characteristic polynomial $P(s)$ of the linearized system \eqref{eq:hat_G} is
$$
P(s) = \det(sI-\hat G) = 0,
$$
which is a polynomial whose coefficients depend on the choice of parameters $\mathbf{p}$ for the digital controller \eqref{eq:input} in either of the representations \eqref{eq:digital_time}-\eqref{eq:digital_state}.
As we have shown in Theorem \ref{thm:instability}, if this polynomial has a root $s$ outside unit circle, \ie, $\|s\|>1$,
then the closed-loop sampled-data nonlinear system is unstable, so we can eliminate that controller from the synthesis domain.

\section{Digital Controller Synthesis}
In this section we define the digital controller synthesis problem for SDSSs, using the same notation of
Definition \ref{defn:sdss}.
Recall that we consider controller parameters $\mathbf{p}\in\mathbb{P}$ (where $\mathbb{P}$ is a hyperbox in $\mathbb R^{2L+1}$), a finite time horizon $T$, and a safety {\em invariant} $\phi(x)$ defined as a predicate over the SDSS state vector $x$ (a quantifier-free FOL formula over the theory of nonlinear real arithmetic).
The synthesized controller should ensure safety with respect to the probability measure induced by the stochastic
disturbance $\mathbf{d}(\cdot)$ and the measurement noise $\eta$. In particular, the probability of satisfying $\phi(x(t))$ for all $t\in[0,T]$ should be above a given, user-defined threshold $\vartheta$.
We further constrain the controller search by limiting the maximum controller degree to $L$ (see \eqref{eq:digital_time} in Definition \ref{defn:sdss}).

Below, we denote by $\theta$ the stochastic uncertainty due to the SDSS disturbances $\mathbf{d}$ and the measurement noise $\eta$ up to time $T$.

\begin{definition}[Digital Controller Synthesis]\label{defn:synth}
Given an SDSS, a time bound $T<\infty $, a maximum controller degree $L$, and a probability threshold $\vartheta\in(0,1)$,  the digital controller synthesis problem is finding the degree $l^*$
\begin{equation*}\label{prob:synth}
	l^* = \text{min} \{\ l\ | \  C(l) \neq \emptyset, l \leq L \}
\end{equation*}
and controller parameters $\mathbf{p}^*\in C(l^*)$ where
\[
C(l) := \left\{ \mathbf{p}\in \mathbb{P}_l \ | \
	\underset{\theta}{\text{Prob}}\left\{\forall t\in[0,T]\;\; \phi(x(\mathbf{p},\theta,t)) \right\}\ge\vartheta
	\right\}
\]
and $\mathbb{P}_l$ is the parameter space of controllers of degree $l$. For clarity we indicate that $x$ depends (indirectly) on the controller parameters $\mathbf{p}$ and on the
stochastic uncertainty $\theta$. If $C(l) = \emptyset$ for all $l\leq L$, then we say the problem is {\em infeasible}.
\end{definition}

In general the above synthesis problem is very hard to solve exactly due to the presence of
nonlinearities, ordinary differential equations (ODEs) introduced by the SDSS dynamics, and multi-dimensional integration
for computing probabilities. In fact, a decision version of the digital controller synthesis problem (\ie, given
$l\leq L$ decide whether $C(l)$ is nonempty), is easily shown to be undecidable.
While Satisfiability Modulo Theory (SMT) approaches, \eg, \cite{DBLP:conf/lics/GaoAC12}, can now in principle handle
nonlinear arithmetics via a sound numerical relaxation, their computational complexity is exponential in
the number of variables.
In particular, our stochastic optimization problem is high-dimensional and currently infeasible for fully formal approaches,
but it can be tackled using the mixed SMT-statistical approach of \cite{DBLP:conf/hvc/ShmarovZ16}, which computes statistically and numerically sound confidence intervals. As such, we replace (exact) probability with empirical mean, thereby obtaining a Monte Carlo version of Definition~\ref{defn:synth}.


\begin{definition}[Digital Controller Synthesis - Monte Carlo]\label{defn:synth_mc}
Given an SDSS $H$, a time bound $T<\infty $, a maximum controller degree $L$, a probability threshold $\vartheta\in(0,1)$, and a confidence value $c\in(0,1)$.
The Monte Carlo digital controller synthesis problem is finding
\begin{equation*}
\label{prob:synth_mc}
	\hat{l}^* = \text{min} \{\ l \ | \  \hat{C}(l) \neq \emptyset, l\leq L \}
\end{equation*}
and controller parameters $\mathbf{\hat p}^*\in \hat{C}(\hat{l}^*)$ where
\[
\hat{C}(l) := \left\{ \mathbf{p}\in \mathbb{P}_l \,\, | \,\, \text{Prob}_{\boldsymbol{\theta}_K} \{ \bar X(\mathbf{p},\boldsymbol{\theta}_K) \ge\vartheta\}\ge c \right\},
\]
and $\boldsymbol{\theta}_K = (\theta_1, \ldots, \theta_K)$ is a finite-dimensional random vector in which each $\theta_i $ is independent and identically distributed from $\theta$,
$\text{Prob}_{\boldsymbol{\theta}_K}$ is the product measure of $K$ copies of the probability measure of $\theta$, 
and $\bar X(\mathbf{p},\boldsymbol{\theta}_K) =
	\frac{1}{K} \sum_{i=1}^K I\left\{ \forall t\in [0,T]\;\; \phi(x(\mathbf{p},\theta_i,t))\right\}$,
with $I\{\cdot\}$ being the indicator function.
\end{definition}
Note that by the law of large numbers when $K\rightarrow \infty$ and the confidence value $c$ is sufficiently close to one, 
we have that $\hat{C}(l)$ approximates $C(l)$ arbitrarily well.
Also, note that the constraints in Definition \ref{prob:synth_mc} are simpler to solve than those in \ref{prob:synth}, as they do not involve absolutely precise integration of the probability measure --- we only require to decide the constraints with some statistical confidence $c$.
However, even with this simplification a decision version of the Monte Carlo digital controller synthesis problem (\ie, deciding whether $\hat{C}(l) = \emptyset$) remains undecidable when plants with nonlinear ODEs are involved. Intuitively, that is because evaluating the elements of $\hat{C}(l)$ amounts to solving
reachability, which is well known to be an undecidable problem for general nonlinear systems.
Hence, one can only solve the Monte Carlo controller synthesis problem {\em approximately}, and that is what we aim to do in the next Section.
Finally, note that we can generate sample realizations of $\theta$: recall from Section~\ref{sec:sdss} that the disturbance $\mathbf{d}(\cdot)$ is defined from a finite number of random variables. This implies that $\theta$ is finite-dimensional as the number of random noise variables is also finite due to the sampling period.

\section{Synthesis Algorithm}\label{sec:method}

In this section we present an algorithm for approximately solving the Monte Carlo controller synthesis
problem of Definition \ref{prob:synth_mc}. Our synthesis algorithm starts from controllers with degree
$l=0$ and iteratively increase $l$ until the constraint $\hat{C}(l) \neq \emptyset$ is satisfied or $l$
reaches a maximum value.

The synthesis algorithm, summarized in Algorithm~\ref{alg:main},
consists of two nested loops.
%
The inner loop (lines \ref{alg:main-inner}-9) consists of two main stages: optimization and verification.
Procedure {\bf optimize} (line \ref{alg:optimize}) aims at finding controller parameters $\mathbf{p}$ that
(approximately) maximizes
the empirical probability that the closed-loop system with a discrete-time version of the plant satisfies property $\phi$ over
the finite time horizon $[0,T]$;
{\bf optimize} also returns an approximate confidence interval (CI) $[a,b]$ for such probability. Optimization is based on the cross-entropy algorithm, but our approach could work with different black-box optimization algorithms too, such as Gaussian process optimization~\cite{rasmussen2004gaussian} and particle swarm optimization~\cite{kennedy2011particle}.
Then, procedure {\bf verify}
(line \ref{alg:verify}) checks the candidate controller $\mathbf{p}$ in closed-loop
with the {\em original} (continuous-time) plant model and computes a precise CI $[a',b']$ for
$\text{Prob}_{\boldsymbol{\theta}_K} \{ \bar X(\mathbf{p},\boldsymbol{\theta}_K) \ge\vartheta\}$.
The reason for using the continuous-time plant only in {\bf verify} is due to the high computational complexity of 
validated numerical ODE solving compared to solving its discrete-time approximation.
The interval returned by {\bf verify} is compared against the current best verified interval, which is then updated
accordingly (line \ref{line:main-result-update}).

The procedures {\bf optimize} and {\bf verify} are iterated until the approximate CI (for the discrete-time plant)
$[a,b]$ and the verified CI (for the continuous-time plant) $[a', b']$ overlap
to a certain length, or a maximum number of iterations is reached (line 9). In the outer loop, if $a' \geq \vartheta$
(\ie, the LHS of the verified confidence interval is larger than $\vartheta$) then $\mathbf{p}$ is a witness
for $\hat{C}(l) \neq \emptyset$, with probability at least $c$. Therefore, $\mathbf{p}$ approximately solves
the digital controller synthesis problem of Definition \ref{prob:synth_mc}, and the algorithm terminates. (The approximation
lies in the fact that we cannot guarantee that the synthesized controller has minimum degree.)
Otherwise ($a' < \vartheta$), we increase the controller degree $l$ up to a maximum degree $L$.


In the inner loop, line \ref{line:main-improve} improves the approximation of the closed-loop
system used in {\bf optimize}. This can be any adjustment to the ODE solver complexity (\eg, increasing the
Taylor series order). In our case, it corresponds to increasing the number of time points used for ODE integration.
We next explain both {\bf optimize} and {\bf verify} in more detail.


\begin{algorithm}[h!]
\RestyleAlgo{ruled}
\LinesNumbered
\DontPrintSemicolon
\SetKwComment{tcz}{\{}{\}}
\SetKwInOut{Input}{Input}\SetKwInOut{Output}{Output}
\SetCommentSty{textit}

\Input{$S$ -- SDSS, $L\ge 0$ -- maximum controller degree,\\
		$\mathbb{P} \!= \![c_0,d_0]\!\times\!\cdots\!\times\! [c_{2L},d_{2L}]$ -- parameters domain, \\
		$\vartheta$ -- probability threshold, \\
		${\bf m}$ -- initial solver discretization,\\
		$\alpha\in(0,1)$ -- factor for tuning ${\bf m}$ (interval overlap),\\
		$\xi$ -- confidence interval size, $c$ -- confidence value.
        }
\Output{$\{{\bf p}^*, [a^*,b^*]\}$ - best performing controller.}
$[a^*,b^*] := [0,0]$; $l:=0$;\;
\Repeat{$l>L$ or $a'\ge \vartheta$}
{
	$\mathbb{P}_l := [c_0,d_0]\times \dots \times [c_{2l}, d_{2l}]$;\label{line:main-complexity}\;
	\Repeat(\label{alg:main-inner}){$|[a,b] \cap [a',b']| \ge \alpha(b-a)$ or $\mathit{MAX\_ITERATIONS}$}
	{
		$\{{\bf p}, [a,b]\} := {\bf optimize}(S,\mathbb{P}_l,{\bf m},\xi,c)$;\label{alg:optimize}\;
		$[a', b'] := {\bf verify}(S, {\bf p}, \xi, c)$;\label{alg:verify}\;
		\lIf{$\frac{a'+b'}{2} > \frac{a^*+b^*}{2}$}{${\bf p}^*:={\bf p}; [a^*,b^*]:=[a',b']$;\label{line:main-result-update}}
		${\bf m} := {\bf update\_discretization(m)}$\label{line:main-improve};
	}
    $l = l+1$;
}
\Return $\{{\bf p}^*, [a^*,b^*]\}$;\;
\caption{Main Synthesis Algorithm}
\label{alg:main}
\end{algorithm}



\paragraph{{\bf Procedure optimize}}
%
%
%
%

In Algorithm~\ref{algo:optimize_func} we give the pseudocode for the \textbf{optimize} function (line~\ref{alg:optimize}
of Algorithm~\ref{alg:main}).
It implements a modified cross-entropy (CE) optimization algorithm \cite{Rubinstein1999} that repeatedly samples
(from the CE distribution of) controller parameters, evaluates their performance, and guides the search towards
parameters that increase the safety probability (\ie, probability of satisfying $\phi$ over $[0,T]$).
Sample performance is computed first by the stability check using Theorem \ref{thm:instability}
(line \ref{alg_optimize:stability} of Algorithm~\ref{algo:optimize_func}). If a controller does not pass the test, \ie, it is necessarily unstable, it is rejected.
Otherwise we compute a CI for the probability (over $\theta$) of the closed-loop system to satisfy
the invariant $\phi$ (line \ref{alg_optimize:CI}).
For this purpose, we consider a discrete-time version of the plant, simulated via an approximate ODE solver based on the first term of the Taylor series expansion.
The time steps $0 = t_1 < \cdots < t_G =T$ for ODE integration are obtained by discretizing the time between the controller sampling points
(defined by $\tau$ --- see Definition \ref{defn:sdss}) using the discretization parameter $\mathbf{m}$.

To compute the CI, our
implementation uses sequential Bayesian estimation for efficiency reasons~\cite{ProbReach}, but other standard statistical techniques may also be employed
(\eg, the Chernoff-Hoeffding bound).
After an adequate number of controller parameters are sampled and evaluated, the best performing sample is chosen (line \ref{alg_optimize:max}), and the CE distribution is updated
accordingly (line \ref{alg_optimize:update_CE}, see \cite{Rubinstein1999} for more details). This is repeated until a maximum number of iterations is reached.

\begin{algorithm}[h!]
\RestyleAlgo{ruled}
\LinesNumbered
\tcp{Modified Cross-Entropy (CE) algorithm}
\DontPrintSemicolon
$\mathbf{p}^* = \bot$; \ $[a^*,b^*] = [0,0]$\\
\Repeat{$MAX\_ITERATIONS$}{
		$Q := \{(\mathbf{p}^*, [a^*,b^*])\}$ \tcp*[f]{sample performance queue}\\
	\Repeat{$\mathit{MAX\_SAMPLES}$}{
		${\bf p}$ := sample controller parameters from CE distribution\\
		\If(\tcp*[f]{Theorem~\ref{thm:instability}}){(${\bf p}$ passes stability check)}{\label{alg_optimize:stability}
		$[a,b]$ := confidence interval, with size $\xi$ and confidence $c$, for probability of satisfying $\phi$ with plant discretization {\bf m}
			and controller $\mathbf{p}$\label{alg_optimize:CI}
		}
		\lElse{	$[a,b] := [0,0]$ }
		$Q := Q \cup \{(\mathbf{p}, [a,b])\}$ \tcp*[f]{add sample performance}\\
	}
	$ (\mathbf{p}^*, [a^*,b^*]) := 
    {\argmax} \left\{(a+b)/2\,\,|\,\,(\mathbf{p}, [a,b])\in Q\right\}$\label{alg_optimize:max}\\
	update CE distribution using $tail(Q)$	\label{alg_optimize:update_CE}\tcp*[f]{discard head}\\
	$Q:=\emptyset$		\tcp*[f]{empty queue}\\
}
\Return $\{(\mathbf{p}^*, [a^*,b^*])\}$
\caption{$\{{\bf p}, [a,b]\} := {\bf optimize}(S,\mathbb{P}_l,{\bf m},\xi,c)$}
\label{algo:optimize_func}
\end{algorithm}

%

\paragraph{\bf{Procedure verify}}
We use the ProbReach tool \cite{ProbReach} to compute a CI for the probability
that a candidate digital controller in closed-loop with the plant satisfies
the time-bounded invariant $\phi$ (line \ref{alg:verify} of Algorithm \ref{alg:main}). This step is necessary
since the candidate controller has
been obtained using an approximate, discrete-time solver for simulating the (continuous) plant dynamics,
while ProbReach uses instead an SMT solver~\cite{dReal} to handle the plant dynamics in a sound manner. 
In particular, 
ProbReach allows to derive a guaranteed confidence interval for
$\text{Prob}_{\boldsymbol{\theta}_K} \{ \bar X(\mathbf{p},\boldsymbol{\theta}_K) \ge\vartheta\}$ with confidence $c$,
where $K$ is the number of samples of the Monte Carlo synthesis problem (see Definition \ref{defn:synth_mc}).
(We note that in general $K$ will depend on the size of the interval and the confidence $c$.)

Procedure {\bf verify} consists of two steps. The first step builds a hybrid system (the model format accepted by ProbReach)
representing the closed-loop system under the candidate controller.
In the second step, {\bf verify} invokes ProbReach with three parameters: the hybrid system, and the required minimum size $\xi$ and confidence $c$
of the confidence interval to compute.
We remark that the size of the confidence interval cannot
be guaranteed in general \cite{DBLP:conf/hvc/ShmarovZ16} because of the undecidability of
reasoning about nonlinear arithmetics. As such, the confidence interval returned by ProbReach via {\bf verify}
can be fully trusted from both the statistical and numerical viewpoints: while the interval size might be larger than $\xi$,
the confidence is guaranteed to be at least $c$, as the sampled controllers are evaluated by SMT and verified numerical techniques.

\begin{theorem}
Let $S$ be an SDSS for which the synthesis problem of Definition \ref{defn:synth} is feasible for a given $\vartheta\in(0,1)$
and controller degree $l\leq L$. Suppose that Algorithm \ref{alg:main}, with parameters $S, L, \vartheta$ and $c\in (0,1)$,
returns a controller $\mathbf{p}$ and an interval $[a,b]$ such that $a \geq \vartheta$. Then $\mathbf{p}$ is a solution of
the Monte Carlo problem of Definition \ref{defn:synth_mc} for $S, \vartheta$, and $c$.
\end{theorem}
\begin{proof}[Proof (sketch)]
It suffices to note that Algorithm \ref{alg:main} terminates either by finding
a controller of minimal degree whose safety probability is at least $\vartheta$, with confidence $c$, or by finding
a controller of degree $l$ with the best confidence interval produced by {\bf verify}.
By hypothesis a controller of degree $l$ that satisfies $\phi$ with probability larger than $\vartheta$ exists, and
Algorithm \ref{alg:main} returns an interval whose LHS is larger than $\vartheta$. Therefore, by {\bf verify},
we know that this interval has confidence $c$.
\end{proof}
We remark that Algorithm \ref{alg:main} can output sub-optimal controllers: this is unavoidable when using stochastic
optimization methods such as cross-entropy. However if the controller design problem is feasible, multiple restarts of
Algorithm \ref{alg:main} will eventually find the minimal controller with probability 1.

\section{Case studies and evaluation}\label{sec:evaluation}
We evaluate our approach on three case studies: a model of insulin control for Type 1 diabetes (T1D)~\cite{hovorka2011closed}, also known as the artificial pancreas (AP),  a model of a powertrain control system~\cite{jin2014powertrain}, and a quadruple-tank process. For all case studies we use the following input parameters
for Algorithm \ref{alg:main}: $\xi = 0.05$, $c = 0.99$ and $\alpha=0.5$. The experiments were performed on a 32-core Intel 2.90GHz system running Ubuntu.

\paragraph{\bf{Digital PID controllers}.}
While our algorithm can synthesize any digital controller as per Definition \ref{defn:digital_controller}, we here exemplify
its use via proportional-integral-derivative (PID) controllers, one of the most popular control techniques. 
A PID controller output is the weighted sum of three terms: the error itself weighted with $K_P$, its rate of change weighted with $K_D$, and accumulated error weighted with $K_I$.
The input/output equation of a digital PID controller is
\begin{align}
u(k) & = u(k-1) + K_P \left[e(k)-e(k-1)\right]\nonumber\\
& + K_I e(k) + K_D[e(k)-2e(k-1)+e(k-2)]\ .
\label{eq:PID_time}
\end{align}
Essentially, the controller needs to store the previous value of the input and the previous two values of the error. In the following case studies, we focus on the synthesis of controllers in the PID form, hence we consider a maximum degree $L=2$.

\subsection{Artificial pancreas}
The AP is a system for the automated delivery of insulin therapy that is required to keep blood glucose (BG) levels of diabetic patients within safe ranges, typically between 4-11 mmol/L. A so-called continuous glucose monitor (CGM) sends BG measurements to a control algorithm that computes the adequate insulin input. PID control is one of the main techniques~\cite{steil2011effect}, and is also found in commercial devices~\cite{kanderian2014apparatus}.

Meals are the major disturbance in insulin control, which make full closed-loop control challenging. Our approach is therefore well suited to solve this problem because it can synthesize controllers attaining arbitrary safety probability by minimizing the impact of such disturbances.
To model insulin and glucose dynamics, we employ the nonlinear model of Hovorka {\em et al}.~\cite{Hovorka04}, considered as one of the most faithful models. The plant has nine state variables describing insulin and glucose concentration in different physiological compartments.
We evaluate the system for a time bound of 24 hours.

In our SDSS model, we consider three meals (respectively represent breakfast, lunch and dinner) with random timing and random amount, expressed by the following normally-distributed parameters: the amount of carbohydrates (CHO) of each meal in grams,
$D_{G_0}\sim\mathcal{N}(50,100)$, $D_{G_1}\sim\mathcal{N}(70,100)$ and $D_{G_2}\sim\mathcal{N}(60,100)$,
and the waiting times between meals,
$T_1 \sim \mathcal{N}(300,100)$ and $T_2 \sim \mathcal{N}(300,100)$.
The corresponding disturbance input is given by:
$$\mathbf{d}(t) = \{ D_{G_0} \text{ if } t = 0; \ D_{G_1} \text{ if } t = T_1; \ D_{G_2} \text{ if } t = T_2; \ 0 \text{ otherwise} \}.$$



The system output $y(t)$ is the CGM measurement (performed every 5 minutes), given by the equation $y(t) = C(t) + \eta(t)$, where $C$ is the state variable for interstitial glucose and $\eta(t)$ is white Gaussian sensor noise with standard deviation $0.25$.

The control input $u(t)$ is the insulin infusion rate computed by the PID controller. The tracking error is defined as $e(t) = r(t) - y(t)$ with the constant reference signal $r(t) = 6.11$ mmol/L.
The total infusion rate is given by $u(t) + u_b$ where $u_b$ ($\approx 0.05548$) is the basal insulin, i.e., a low and continuous dose to regulate glucose outside meals. The value of $u_b$ is chosen to guarantee a steady-state BG value equals to $r(t)$ in absence of meals. This steady state is used as the initial state of the system.

\paragraph{Safety property} Insulin control seeks to prevent \textit{hyperglycemia} (BG above 11 mmol/L) and \textit{hypoglycemia} (BG below 4 mmol/L). Hypoglycemia happens when the controller overshoots its dose and has more severe health effects than hyperglycemia, which is tolerated to a small extent after meals. For this reason we consider a safe BG range of $[4,16]$ mmol/L, which strictly avoids hypoglycemia and allows for some post-meal hyperglycemia tolerance. In addition, we want that the glucose level stays close to the reference signal towards the end of the 24 hours (1440 minutes). Our invariant is given by:
$$G \in [4,16] \wedge (t \in [1410, 1440] \rightarrow G \in [r(t)-0.25,r(t)+0.25]),$$
where $G$ is the state variable for the BG concentration.

In the synthesis algorithm, we use a probability threshold of $\vartheta=0.95$ (we want to satisfy the above invariant with probability not below 95\%), confidence $c=0.99$, confidence interval size $\xi = 0.05$, and $\alpha = 0.5$.


\paragraph{Synthesis results} Table~\ref{tbl:PID_AP} shows the PID controllers synthesized at each iteration of the algorithm. The domain of controller parameters was chosen as follows: $K_p \in [-10^{-2}, 10^{-3}]$,
$K_i \in [-10^{-5}, 10^{-6}]$ and $K_d \in [-1, 10^{-1}]$. Even though none of the synthesized controllers achieves the probability threshold $\vartheta = 0.95$, the degree-2 controller (PID) is very close to satisfying the property, with a $99\%$-confidence interval of $[0.94242,0.99242]$. 


\begin{table}
\centering
\begin{footnotesize}
\setlength{\tabcolsep}{0.5em}
\begin{tabular}{c|cccc|ccc}
$l$ & $K_P$ \tiny{$\times 10^3$} & $K_I$ \tiny{$\times 10^7$} & $K_D$ \tiny{$\times 10$} & $[a^*, b^*]$ & $\#_o(\#^0_o)$ & $\#_c(\#_{un})$ & CPU(opt)\\
\hline
0 & -5.006 & - & - & [0.938,0.988] & 8(1) & 165(0) & 1130(586)\\
1 & -5.4 & -2.179 & - & [0.939,0.989] & 8(8) & 217(0) & 1534(873)\\
2 & -5.716 & -1.88 & -2.002 & [0.942,0.992] & 8(8) & 301(0) & 2100(1312)\\
\hline
\end{tabular}
\end{footnotesize}
\caption{Controller synthesis for the artificial pancreas system. $l$ -- controller degree,
$K_P$, $K_I$, $K_D$ -- controller gains, $[a^*,b^*]$ -- confidence interval for safety probability (with $c=0.99$), $\#_o(\#^0_o)$ -- number of points used by the non-verified ODE solver at the end (beginning) of iteration,
$\#_c(\#_{un})$ -- number of candidates (unstable) sampled by the optimization algorithm,
CPU(opt) -- total (only {\bf optimize} procedure) runtime in minutes.}
\vspace*{-.6cm}
\label{tbl:PID_AP}
\end{table}

To better understand the performance of the controllers, we analyze their behavior on 1,000 Monte Carlo executions of the system. Results, reported in Figure~\ref{fig:sims_AP}, evidence that hyper- and hypo-glycemia episodes are never sustained. 


\begin{figure}
\centering
\vspace*{-.4cm}
\subfloat[$l$=0]{\includegraphics[width=.33\columnwidth]{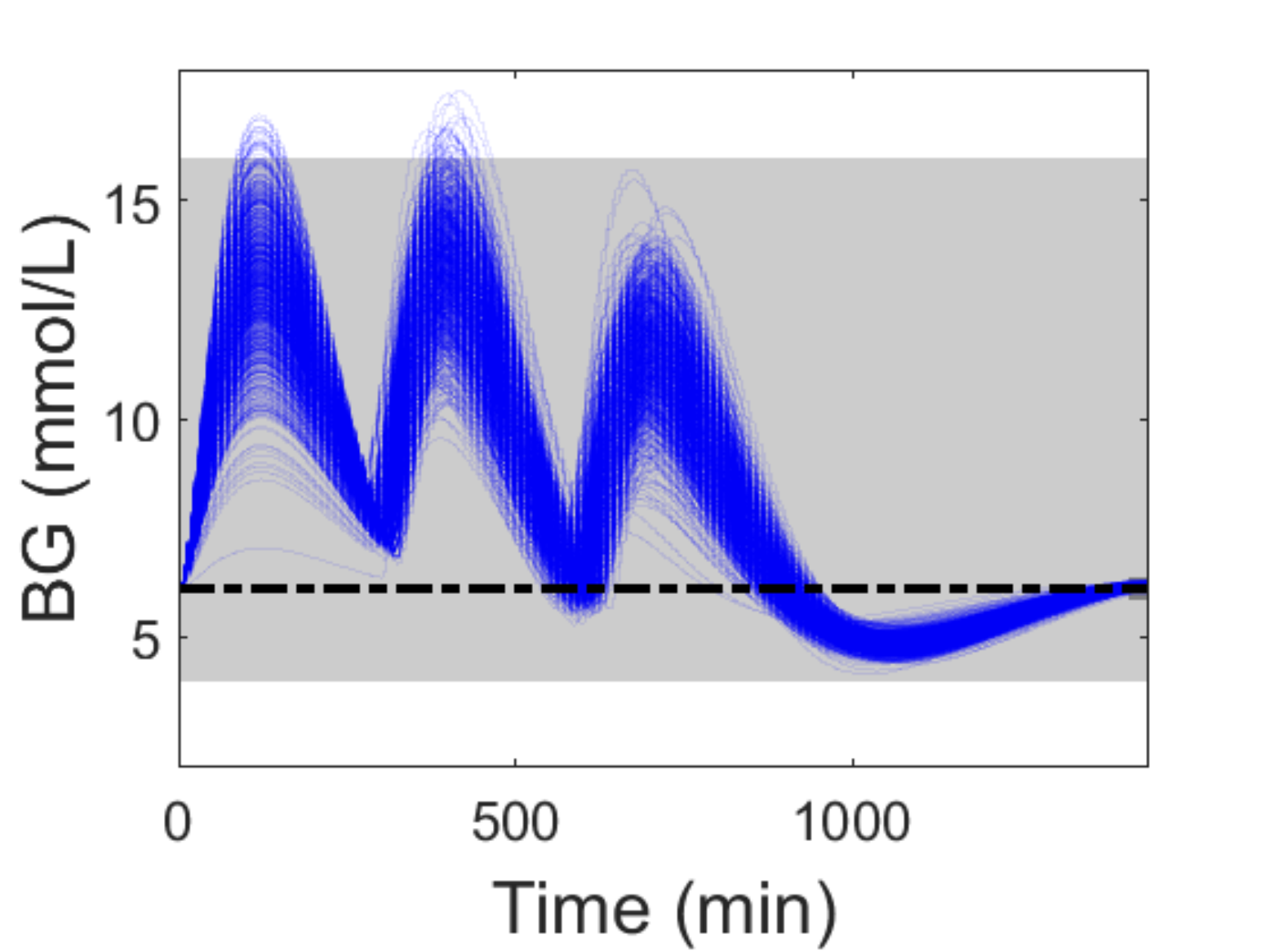}}
\subfloat[$l$=1]{\includegraphics[width=.33\columnwidth]{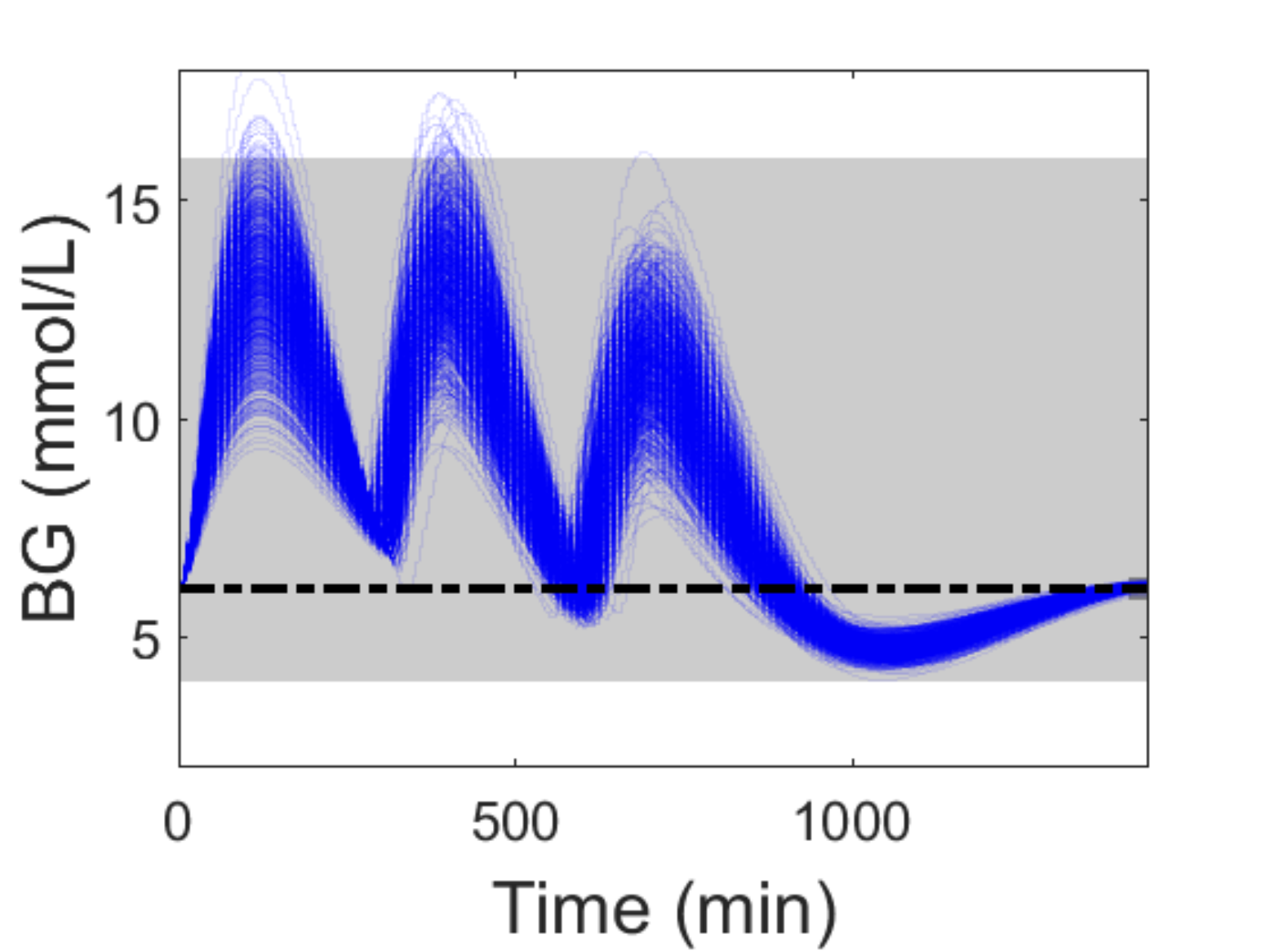}}
\subfloat[$l$=2]{\includegraphics[width=.33\columnwidth]{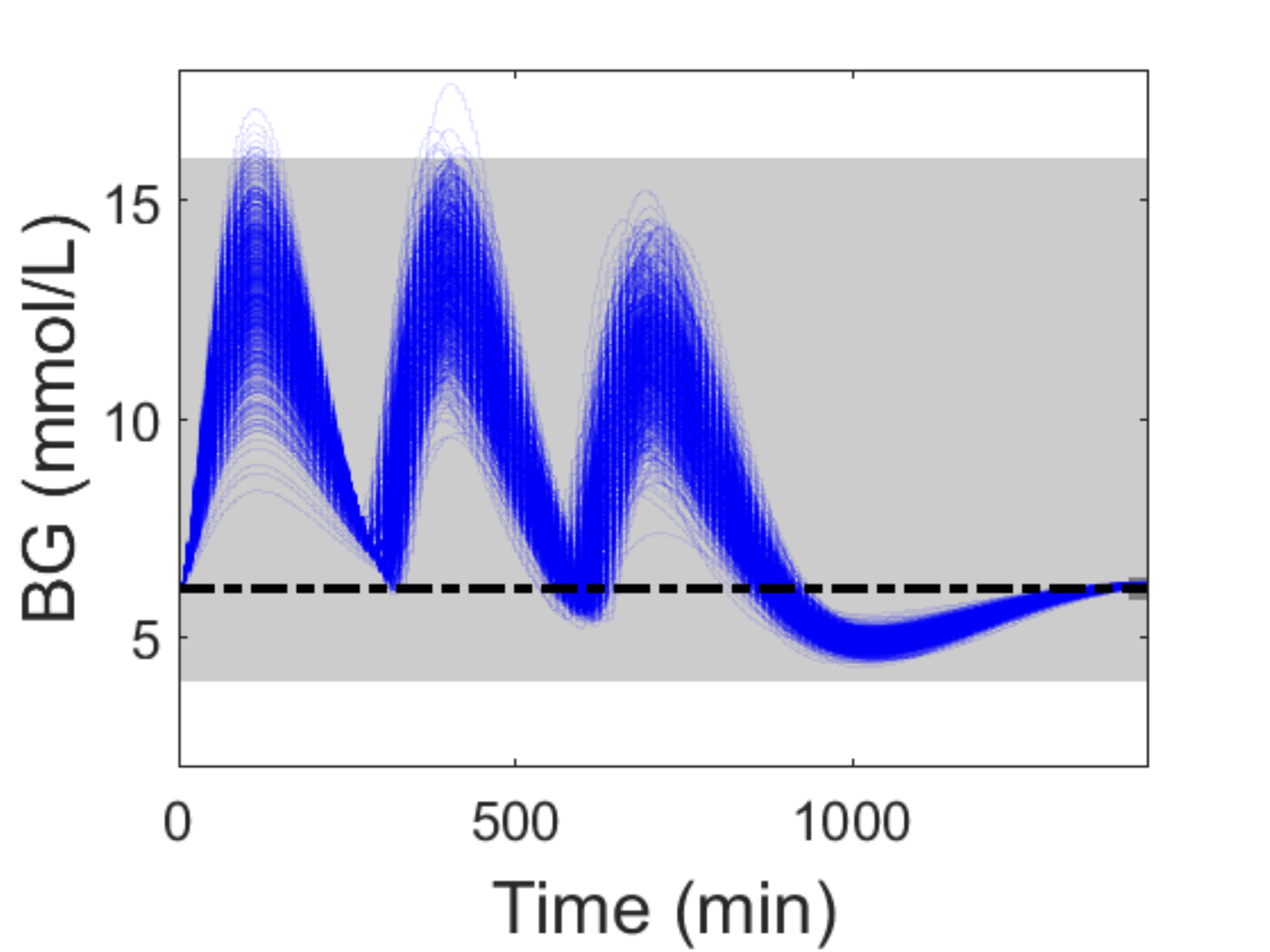}}\\[3pt]
\vspace*{-.3cm}
\caption{Evaluation of synthesized controllers (degrees $l=0,1,2$) on 1,000 simulations of the AP system. Blue lines: BG profiles; grey areas: healthy BG ranges ($G \in [4,16]$); dashed black lines: reference BG.}
\vspace*{-.4cm}
\label{fig:sims_AP}
\end{figure}





\subsection{Powertrain system}
We consider the automotive air-fuel control system adapted from the powertrain control benchmark in~\cite{jin2014powertrain}. The plant model consists of a system of three nonlinear ODEs describing the dynamics of the engine in relation to the throttle air dynamics, intake manifold and air-fuel path.

The system has two exogenous inputs, captured by the disturbance vector $\mathbf{d}(t) = [\omega(t) \ \theta_{in}(t)]^T$, where $\omega$ (rad/s) is the engine speed
($\omega \sim \mathcal{N}(105,4)$),
and $\theta_{in}$ (degrees) is the throttle angle. $\theta_{in}(t)$ is defined as a pulse train wave with random amplitude
$a \sim \mathcal{N}(30.6,25)$
and period $\zeta = 4$:
$$\theta_{in}(t) = a \,I \{t\in [0,\zeta/2)\} + 8.8 \,I\{t\in [\zeta/2,\zeta]\}.$$


The noisy plant output is $y(t) =  \lambda(t) + \eta(t)$, where $\lambda(t)$ is the air/fuel ratio, and $\eta(t) \sim \mathcal{N}(0,0.0625)$.

The engine is controlled by a PID controller that seeks to maintain a constant air/fuel ratio equals to the stoichiometric value $\bar{\lambda} = 14.7$, that is when the engine performs optimally. The tracking error is thus given by $e(t) = y(t) - \bar{\lambda}$. The control signal $u(t)$ determines the amount of fuel entering the system. 

\paragraph{Safety property} We consider the following invariant
\small
$$|\mu(t)| < 1 \wedge (t \in [\zeta/8,\zeta/2] \cup [5\zeta/8,\zeta]) \rightarrow |\mu(t)| < 0.05)$$
\normalsize
where $\mu(t) = (\lambda(t) - \bar{\lambda})/\bar{\lambda}$ is the relative error from the setpoint.
%
The first conjunct states that the air/fuel ratio should constantly be within $\pm 100\%$ of the ideal ratio $\bar{\lambda}$.
The second conjunct states that whenever the input throttle angle $\theta_{in}$ rises (at time $t=0$) or
falls ($t=\zeta/2$), the plant should settle within time $\zeta/8$ and remain in the settling region ($\pm 5\%$ around $\bar{\lambda}$) until the next rise or fall (happening after time $\zeta/2$). We set the probability threshold to $\vartheta=0.96$. 

\paragraph{Synthesis results} Table~\ref{tbl:PID_PT} shows the PID controllers synthesized at each iteration of the algorithm. The domain of controller parameters was chosen as follows: $K_p \in [-0.1, 0.5]$,
$K_i \in [-0.05, 0.2]$ and $K_d \in [-0.05, 0.05]$. With our algorithm, we could synthesize a degree-2 controller (PID) satisfying the threshold. The optimal degree-1 controller has similar performance (both yield confidence intervals with RHS equals to 1), albeit below the threshold. 

\begin{table}
\centering
\begin{footnotesize}
\setlength{\tabcolsep}{0.5em}
\begin{tabular}{c|cccc|ccc}
$l$ & $K_P$ & $K_I$ & $K_D\times 10^3$ & $[a^*, b^*]$ & $\#_o(\#^0_o)$ & $\#_c(\#_{un})$ & CPU(opt)\\
\hline
0 & 0.2713 & - & - & [0.783,0.834] & 128(64) & 74(7) & 1068(944)\\
1 & 0.2004 & 0.0537 & - & [0.954,1] & 128(128) & 134(22) & 1838(1690)\\
2 & 0.2082 & 0.0759 & -4.9551 & [0.963,1] & 128(128) & 214(72) & 2337(2165)\\
\hline
\end{tabular}
\end{footnotesize}
\caption{Controller synthesis for the fuel control system. See caption of Table \ref{tbl:PID_AP}.}
\vspace*{-1cm}
\label{tbl:PID_PT}
\end{table}

Compared to the AP case study, we observe that the powertrain model requires generating (and verifying) fewer candidate parameters. 
At the same time the dynamics of the powertrain system appear more challenging to control as the model requires more ODE integration steps (see column $\#_o(\#^0_o)$ of Tables~\ref{tbl:PID_PT} and~\ref{tbl:PID_AP}).

\subsection{Quadruple-tank process}
We consider a quadruple-tank process adapted from~\cite{QT00}, which consists of four interconnected water tanks. The process is illustrated in Figure~\ref{fig:quad-tank}. This model is an example of a multiple-input and multiple-output (MIMO) system with \emph{multivariable right half-plane zeros} \cite{glad2000control} (such zeros bring performance limitations in control problems). We extended the deterministic model of~\cite{QT00} to include uncertainties in the valve settings and random disturbances in the process that remove water from the tanks. 

The process is controlled in a decentralized fashion, by which two digital controllers are designed for the input-output pairs $(u_1, y_1)$ and $(u_2, y_2)$, where $u_1$ and $u_2$ are the input voltages
for the pumps, and $y_1$ and $y_2$ are the water level measurements obtained as $y_1 = 0.5 \cdot h_1$ and $y_2 = 0.5 \cdot h_2$, where $h_1$ and $h_2$
are the water levels in tanks $1$ and $2$, respectively. In this case study we assume that the pumps can only add water to the tanks (and cannot pump it out).

We consider a scenario where at time 0 and then twice after every minute 
we remove a random amount of water (which is model through
reducing the corresponding water levels by a random value $\sim U(0,3)$) from tanks 1 and 2.  Every time such a disturbance happens, the valves parameters are randomly reset to 
$\gamma_1 \sim N(0.7, 0.223)$ and $\gamma_2 \sim N(0.6, 0.223)$. The system is subject to a measurement noise modeled as a white Gaussian noise with variance 0.33.

\paragraph{Safety property}
After each disturbance, we require that the system reaches the desired water levels in tanks 1 and 2 (within 1 centimeter above or below the corresponding set points $r_1 = 12.4$ and $r_2 = 12.7$) within 5 seconds, and that the water levels stay close to the setpoints for the remaining 55 seconds, before the next disturbance occurs. Also, all four water levels $h_1, h_2, h_3, h_4$ must always stay in the interval $[0,20]$ and the input voltages $u_1, u_2$ for both pumps must be in the range $[0,24]$.


\begin{figure}
\centering
\vspace*{-.4cm}
\includegraphics[width=0.5\columnwidth]{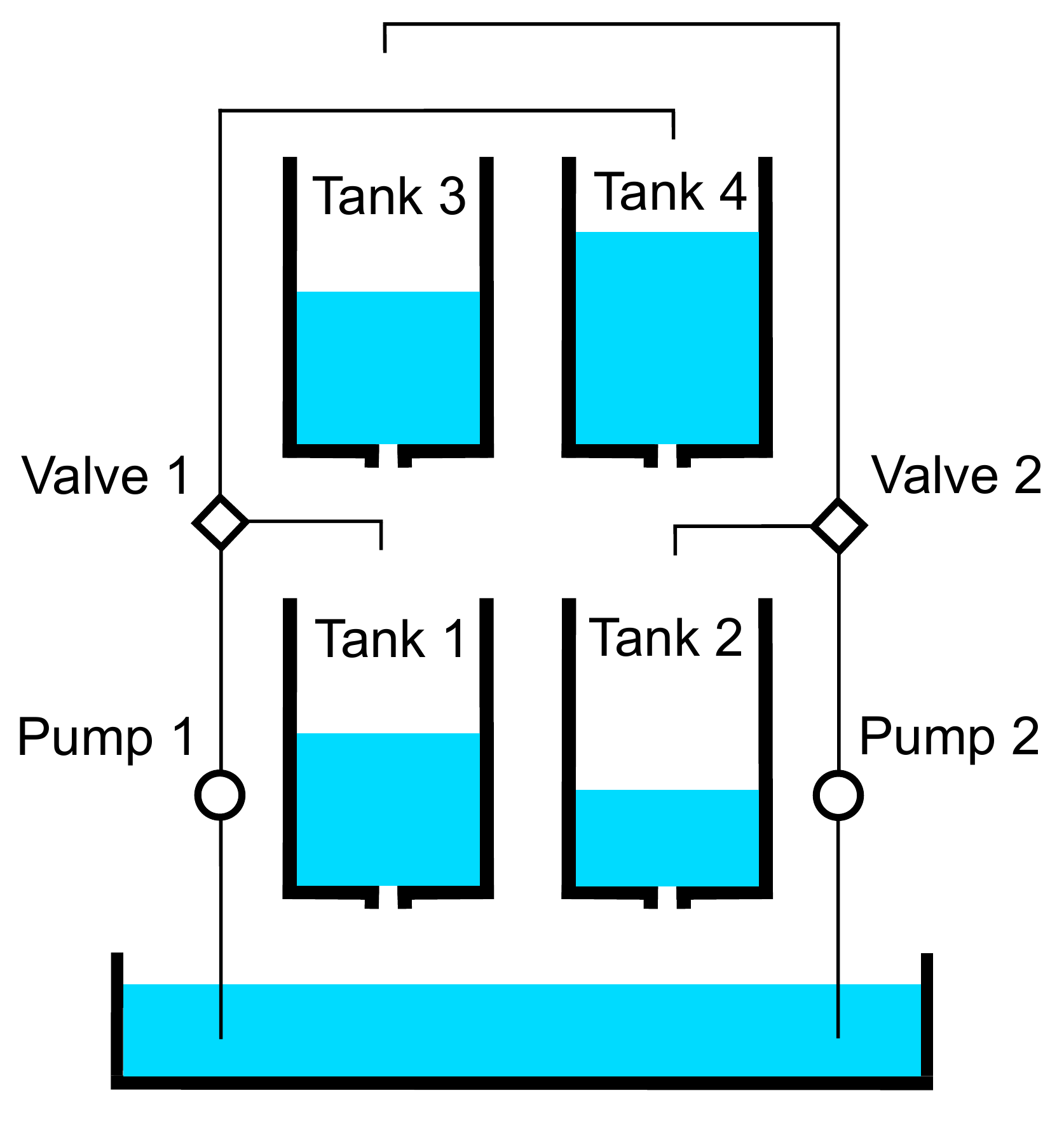}
\vspace*{-.3cm}
\caption{The diagram of the quadruple-tank model.}
\vspace*{-.4cm}
\label{fig:quad-tank}
\end{figure}

\paragraph{Synthesis results} The domain of controller parameters
was chosen as $K_{P_1} \in [-1, 20]$, $K_{I_1} \in [-1, 10]$, $K_{D_1} \in [-1,10]$, $K_{P_2} \in [-1, 20]$, $K_{I_2} \in [-1, 10]$, $K_{D_2} \in [-1, 10]$. The controller synthesis results are presented in Table~\ref{tbl:PID_QT}, which shows that we can obtain a confidence interval of up to $[0.94,0.99]$ for the safety probability by using two PI controllers (see third row of Table~\ref{tbl:PID_QT}). Note that the performance of the controller is not improved by including the derivative terms $K_{D_1},K_{D_2}$ (see last two rows). 
This is due to the optimization algorithm which works by sampling a finite number of controller parameters and thus, might fail to explore parameter regions with better safety probability. 

\begin{table}
\centering
\begin{footnotesize}
\setlength{\tabcolsep}{0.5em}
\begin{tabular}{cccccc|ccc}
$K_{P_1}$ & $K_{I_1}$ & $K_{D_1}$ & $K_{P_2}$ & $K_{I_2}$ & $K_{D_2}$ & $[a^*, b^*]$ & $\#_c(\#_{un})$ & CPU(opt)\\
\hline
10.916 & - & - & 13.085 & - & - & [0.85,0.90] & 58(1) & 164(30)\\
8.463 & 0.650 & - & 12.749 & - & - & [0.89,0.94] & 146(1) & 252(73)\\
6.555 & 1.233 & - & 10.057 & 1.359 & - & [0.94,0.99] & 251(1) & 370(113)\\
6.576 & 1.144 & 1.737 & 9.019 & 1.048 & - & [0.93,0.98] & 507(2)$^*$ & 717(246)\\
6.422 & 1.057 & -0.075 & 6.760 & 1.724 & 3.973 & [0.90,0.95] & 654(2)$^*$ & 917(346)\\
\hline
\end{tabular}
\end{footnotesize}
\caption{Controller synthesis for the quadruple-tank system.
$K_{P_1}$, $K_{I_1}$, $K_{D_1}$, $K_{P_2}$, $K_{I_2}$, $K_{D_2}$ -- controller gains, $[a^*,b^*]$ -- confidence interval for safety probability (with $c=0.99$), $\#_c(\#_{un})$ -- number of candidates (unstable) sampled by the optimization algorithm,
CPU(opt) -- total (only {\bf optimize} procedure) runtime in minutes, $^*$ -- the number of points in discretisation was increased from 1 to 2.}
\vspace*{-.6cm}
\label{tbl:PID_QT}
\end{table}



\section{Related Work}
\label{sec:related}

Although recent papers~\cite{AbateBCCCDKKP17,AbateBCCDKKP17cav,Duggirala015}
have addressed the synthesis of safe digital controllers for linear and 
deterministic systems, synthesis for the class of stochastic nonlinear systems that we consider has not 
yet been tackled.  
 
While the approach proposed in this paper can be applied in principle 
to any digital controller, our case studies focus on PID controllers.
Several methods have been proposed for the synthesis of PID controller for nonlinear and stochastic plants~\cite{su2005design,fliess2013model,guo2005pid,he2011robust,duong2012robust,ShmarovPBLSZ17,TronciTAC2017}. However, none of these methods can provide safety guarantees beside the work of~\cite{ShmarovPBLSZ17} (discussed at the end of the section) .


We pursue a different direction with respect to the classical solutions 
proposed in the literature, by providing a framework for the efficient 
synthesis of discrete-time digital controller for nonlinear and stochastic 
plants that are safe and robust (with respect to a probabilistic 
reachability property) by construction.

The control synthesis problem considered here
requires to solve a parameter synthesis problem 
over a closed-loop system modeled as a stochastic nonlinear system. 
In contrast with existing parameter 
synthesis techniques for stochastic 
 and continuous nonlinear systems~\cite{BartocciBNS15,BortolussiS15,HaghighiJKBGB15}, 
our approach, which targets specifically digital 
controllers, takes advantage of a 
computationally efficient stability check that rules 
out unstable controller candidates, thereby reducing the 
computational effort.  

The problem of controller synthesis under safety requirements has
been investigated mostly for Model Predictive Control (MPC)~\cite{camacho2013model} 
whose goal is to find the control input that optimizes the predicted performance of the closed-loop system up to a finite horizon. The work of~\cite{KaramanSF08,raman2014model,kim2017dynamic,WongpiromsarnTM12,pant2017smooth,LiNSXL17,SadighK16} 
consider safety requirements expressed as temporal logic formulas, 
and synthesize MPC controllers that optimize the robust satisfaction of 
the formula~\cite{donze2010robust} (i.e., a quantitative measure of satisfaction). 
MPC is an online method that requires solving at runtime often computationally 
expensive optimization problems. In contrast, our approach performs controller synthesis at design time.

The closest paper to our own is the work of~\cite{ShmarovPBLSZ17}, where the authors have 
recently proposed a method to synthesize continuous-time 
PID controllers for nonlinear stochastic 
plants such that the resulting closed-loop system satisfies
safety and performance 
requirements specified as bounded 
reachability properties.  
That method works under the assumptions that
the system can measure the output of the plant continuously 
and without sensing noise. However, this is not realistic in the majority 
of embedded systems whose operations 
are governed by a discrete-time clock and where sensor noise 
is unavoidable.

\section{Conclusions}
The synthesis of digital controllers for cyber-physical systems with nonlinear and stochastic dynamics is a challenging problem, and for such systems, no automated methods currently exist for deriving controllers with rigorous and quantitative safety guarantees. In this paper, we have presented a solution to this problem based on two key contributions: a method to check the candidate controllers for the sampled-data nonlinear system with respect to stability; and a two-stage synthesis algorithm that alternates between a fast candidate generation phase (based on Monte-Carlo sampling and non-validated ODE solving) and a verification phase where we derive numerically and statistically valid confidence intervals on the safety probability of the closed-loop system. With this method, we managed to synthesize controllers for three nonlinear systems (artificial pancreas, powertrain, and quadruple-tank process) characterized by large stochastic disturbances and sensing noise. 
As future work, we plan to extend our method to hybrid systems and controllers with fixed-point precision. 


%
\newpage
\appendix
\section{Stability Requirement}\label{app:stability}

The following lemma provides a bound on exponent of a matrix and is is used in proving our main theorem.
\begin{lemma}
\label{lem:matrix_bound}
For any matrix $A$ and any $\epsilon>0$, there is a constant $C(\epsilon)$ such that
\begin{equation*}
\|e^{At}\|_2\le C(\epsilon)e^{(a+\epsilon)t},\quad \forall t\ge 0,
\end{equation*}
with $a = \max_i Re(\lambda_i(A))$.
\end{lemma}
\begin{proof}
We use the Jordan normal form $J$ of $A$ that satisfies $A = P^{-1} J P$ for some invertible matrix $P$:
\begin{align*}
\|e^{At}\|_2 & = \|P^{-1}e^{Jt}P\|_2\le \|P^{-1}\|_2\|P\|_2\|e^{Jt}\|_2\\
& \le \|P^{-1}\|_2\|P\|_2\max_i\|e^{J_it}\|_2,
\end{align*}
where $J_i$ is the $i^{\text{th}}$ block of $J$ with size $k_0\in\mathbb N$ and can be written as $J_i = \lambda_i I + N$. Matrix $N$ is a matrix of all zeros except identities on the first superdiagonal.
Since $N^m = 0$ for all $m> k_0$, we have the following for $J_i$:
\begin{align*}
e^{J_it} & = e^{(\lambda_i I+N)t} = e^{\lambda_i t}e^{Nt}\\
& = e^{(\lambda_i+\epsilon) t}e^{Nt} e^{-\epsilon t} = e^{(\lambda_i+\epsilon) t}\sum_{m=0}^{k_0}\frac{N^m}{m!}t^me^{-\epsilon t}\\
\Rightarrow & \|e^{J_it}\|_2\le e^{(Re(\lambda_i)+\epsilon) t}\left\|\sum_{m=0}^{k_0}\frac{N^m}{m!}t^me^{-\epsilon t}\right\|_2\!\!\le C_i(\epsilon)e^{(Re(\lambda_i)+\epsilon) t}.
\end{align*}
The claim is true by taking $C(\epsilon) := \|P^{-1}\|_2\|P\|_2\max_i C_i(\epsilon)$.
\end{proof}

\begin{proof}[Proof of Lemma \ref{lem1}]
Define the function $W(t) := \|x_1(t)\|^2$,
\begin{align*}
\frac{d}{dt}W(t)&  = 2\left[A x_1(t) + B C_c x^c_1(k) - B D_c C x_1(k\tau) + g(x_1(t),u_1(k\tau))\right]^T\!\! x_1(t)\\
& \le 2\|x_1(t)\|\left[\|A\| \|x_1(t)\| + \|B C_c\| \|x^c_1(k)\| \right.\\
& \qquad\qquad\quad\left. + \|B D_c C\| \|x_1(k\tau)\|+ \gamma\|x_1(t)\| + \gamma\|u_1(k\tau)\|\right].
\end{align*}
We write this inequality in terms of $\sigma(t) := \sqrt{W(t)}$ as
\begin{align*}
2\sigma(t)\frac{d}{dt}\sigma(t) & \le 2(\|A\|+\gamma) \sigma^2(t) + 2\sigma(t)\left( \|B C_c\| + \gamma \|C_c\|\right) \|x^c_1(k)\|\\
& + 2\sigma(t)\left(\|B D_c C\| + \gamma\|D_c C\| + \gamma^2\|D_c\|\right)\sigma(k\tau)
\end{align*}
\begin{align*}
\Rightarrow\frac{d}{dt}\sigma(t) & \le (\|A\|+\gamma) \sigma(t) + \left( \|B C_c\| + \gamma \|C_c\|\right) \|x^c_1(k)\|\\
& +\left(\|B D_c C\| + \gamma\|D_c C\| + \gamma^2\|D_c\|\right)\sigma(k\tau).
\end{align*}
Rename $L := \|A\|+\gamma$, $L_1 := \|B C_c\| + \gamma \|C_c\|$ and $L_2:=\|B D_c C\| + \gamma\|D_c C\| + \gamma^2 \|D_c\|$, to get
\begin{align*}
& \frac{d}{dt}e^{-Lt}\sigma(t) \le e^{-Lt} \left( L_1 \|x^c_1(k)\| + L_2 \|\sigma(k\tau)\|\right)\\
& \Rightarrow
e^{-Lt}\sigma(t) - e^{-Lk\tau}\sigma(k\tau)\le \frac{e^{-Lt} - e^{-Lk\tau}}{-L}\left( L_1 \|x^c_1(k)\| + L_2 \|\sigma(k\tau)\|\right)\\
& \Rightarrow
\sigma(t)\le e^{L(t-k\tau)}\sigma(k\tau) + \frac{1 - e^{L(t-k\tau)}}{-L}\left( L_1 \|x^c_1(k)\| + L_2\|\sigma(k\tau)\|\right)\\
& \Rightarrow
\sigma(t)\le h_1(t-k\tau,\gamma)\sigma(k\tau) + h_2(t-k\tau,\gamma) \|x^c_1(k)\|,
\end{align*}
with functions
\begin{equation}
h_1(t,\gamma) := e^{Lt} + (e^{Lt}-1)L_2/L,\quad h_2(t,\gamma) := (e^{Lt}-1)L_1/L,
\end{equation}
where $L,L_1,L_2$ depend on $\gamma$ as defined above.
\end{proof}


\begin{proof}[Proof of Lemma \ref{lem2}]
Using the assumption of $\|(x_1(t),u_1(k\tau))\|\le r(\gamma)$ for all $t\in[k\tau,(k+1)\tau]$, we get
\begin{align*}
\|\hat g(k\tau)\|& \le \int_{0}^{\tau}\|e^{A(\tau-\lambda)}\|\|g(x_1(k\tau+\lambda),u_1(k\tau))\|d\lambda\\
& \le \gamma\int_{0}^{\tau}\|e^{A(\tau-\lambda)}\|(\|x_1(k\tau+\lambda)\| +\|u_1(k\tau)\|)d\lambda.
\end{align*}
Then we employ Lemma~\ref{lem1} to get
\begin{align*}
\|\hat g(k\tau)\|
\le & \gamma\int_{0}^{\tau}\|e^{A(\tau-\lambda)} \|(h_1(\lambda,\gamma)+\|D_c C\| + \gamma \|D_c\|)\|x_1(k\tau)\|d\lambda\\
& + \gamma\int_{0}^{\tau}\|e^{A(\tau-\lambda)}\|(h_2(\lambda,\gamma)+\|C_c\|) \|x_1^c(k)\|d\lambda.
\end{align*}
The claim holds with any $\hat h_1,\hat h_2$ with
\begin{align*}
\hat h_1(\tau) & \ge  \int_{0}^{\tau}\|e^{A(\tau-\lambda)}\|(h_1(\lambda,\gamma)+\|D_c C\| + \gamma\|D_c\|)d\lambda\\
\hat h_2(\tau,\gamma) & \ge \int_{0}^{\tau}\|e^{A(\tau-\lambda)}\|(h_2(\lambda,\gamma)+\|C_c\|)d\lambda,
\end{align*}
which can be selected as the following according to Lemma~\ref{lem:matrix_bound}, with $\Gamma := C(1)$ and $\alpha := a+1$,
\begin{align*}
\hat h_1&(\tau) =  \int_{0}^{\tau} \Gamma e^{\alpha(\tau-\lambda)}\|(h_1(\lambda,\gamma)+ \|D_c C\| + \gamma \|D_c\|)d\lambda\\
& = \Gamma\left(1+\frac{L_2}{L}\right)\left(\frac{e^{L\tau}- e^{\alpha\tau}}{L-\alpha}\right)+\Gamma\left(\|D_c C\| + \gamma\|D_c\| - \frac{L_2}{L}\right)\frac{e^{\alpha\tau}-1}{\alpha}, \\
\hat h_2&(\tau) =  \int_{0}^{\tau}\Gamma e^{\alpha(\tau-\lambda)}(h_2(\lambda,\gamma)+\|C_c\|)d\lambda\\
 & = \frac{\Gamma L_1}{L}\left(\frac{e^{L\tau}- e^{\alpha\tau}}{L-\alpha}\right)+\Gamma\left(\|C_c\| -\frac{L_1}{L}\right)\left(\frac{e^{\alpha\tau}-1}{\alpha}\right).
\end{align*}

\end{proof}

\section{Gluco-regulatory ODE model}\label{app:AP}

The model consists of three subsystems:
\begin{itemize}
	\item {\em Glucose Subsystem}: it tracks the masses of glucose (in mmol) in the accessible ($Q_1$) and non-accessible ($Q_2$)
    compartments, $G$ (mmol/L) represents the glucose concentration in plasma, $EGP_0$ (mmol/min) is the endogenous glucose production rate and $U_G(t)$ (mmol/min) is the glucose absorption rate from the gut.
    \item  \emph{Gut absorption}: this subsystem uses a chain of two compartments, $G_1$ and $G_2$ (mmol), to model the absorption dynamics of ingested food, given by the disturbance $D_G(t)$. $A_g$ is the CHO bio-availability. $t_{maxG}$ (min) is the time of maximum appearance rate of glucose.
    \item \emph{Interstitial glucose}: $C$ is the subcutaneous glucose concentration (mmol/L) detected by the CGM sensor and has a delayed response w.r.t. the blood concentration $G$.
    \item {\em Insulin Subsystem}: it represents absorption of subcutaneously administered insulin. It is defined by 
    a two-compartment chain, $S_1$ and $S_2$ measured in U (units of insulin),  where $u(t)$ (U/min) is the administration of insulin computed by the 
    PID controller, $u_b$ (U/min) is the basal insulin infusion rate and $I$ (U/L) indicates the insulin concentration in plasma.
    \item {\em Insulin Action Subsystem}: it models the action of insulin on glucose distribution/transport, $x_1$,
    glucose disposal, $x_2$, and endogenous glucose production, $x_3$ (unitless).
\end{itemize}
The model parameters are given in Table
\ref{table:params_AP}. 

\begin{align*} \label{eq:ap-odes}
\dot{Q_1}(t) &= -F_{01} - x_1(t) Q_1(t) + k_{12}\cdot Q_2(t) -F_R +EGP_0(1-x_3) + U_G(t)\\
\dot{Q_2}(t) &= x_1(t) Q_1(t) -(k_{12} + x_2(t)) Q_2(t)\\
G_1(t) &= -\frac{G_1(t)}{t_{maxG}} + A_G \cdot D_G(t) \quad G_2(t) = \frac{G_1(t)-G_2(t)}{t_{maxG}} \\
U_G(t) &= \frac{G_2(t)}{t_{maxG}} \quad G(t) =\frac{Q_1(t)}{V_G} \quad C(t) = k_{\rm int}(G(t) - C(t))\\
\dot{S_1}(t) &= u(t) + u_b - \frac{S_1(t)}{t_{maxI}} \quad \dot{S_2}(t) =\frac{S_1(t)-S_2(t)}{t_{maxI}}\\
\dot{I}(t) &= \frac{S_2(t)}{t_{maxI}\cdot V_I}-k_e I \quad \dot{x_i}(t) =-k_{a_i}\cdot x_i(t)+k_{b_i}\cdot I(t), i = 1,2,3\\
\end{align*}

\begin{table}
\centering
\begin{tabular}{|c|c||c|c||c|c|}
\hline
\bf{par} & \bf{value} & \bf{par} & \bf{value} & \bf{par} & \bf{value} \\
\hline
\hline
$w$ & 100 & $k_e$ & 0.138 & $k_{12}$ & 0.066 \\ 
$k_{a_1}$ & 0.006 & $k_{a_2}$ & 0.06 & $k_{a_3}$ & 0.03 \\ 
$k_{b_1}$ & 0.0034 & $k_{b_2}$ & 0.056 & $k_{b_3}$ & 0.024 \\ 
$t_{maxI}$ & 55 & $V_I$ & $0.12\cdot w$ & $V_G$ & $0.16\cdot w$ \\ 
$F_{01}$ & $0.0097\cdot w$ & $t_{maxG}$ & 40 & $F_R$ & 0 \\
$EGP_0$ & $0.0161\cdot w$ & $A_G$ & 0.8 & $k_{\rm int}$ & 0.025\\
\hline
\end{tabular}
\vspace{2ex}
\caption{Parameter values for the glucose-insulin regulatory model. $w$ (kg) is the body weight.}
\label{table:params_AP}
\end{table}

\section{Fuel Control System Model}\label{app:af_model}
The dynamics of the engine (plant) are given by the following set of ODEs:
\begin{align*}
\dot{p} &= c_1 \left( \dot{m}_{af} - \dot{m_c} \right)\\
\dot{\theta} &= 10(\theta_{in} - \theta)\\
\dot{\lambda} &= c_{26} \left( \frac{\dot{m_c}}{c_{25} F_c} - \lambda \right),
\end{align*}
where $p$ (bar) is the intake manifold pressure; $\theta$ (degrees) is the throttle angle; $\lambda$ is the air/fuel ratio; $\theta_{in}$ (degrees) is the throttle angle input disturbance; $\hat{\theta}$ is the throttle plate angle and is defined by:
$$\hat{\theta} = c_6 + c_7\theta + c_8\theta^2 + c_9\theta^3;$$
$\dot{m}_c$ (g/s) is the air inflow rate to cylinder and is defined by: 
$$\dot{m}_c = c_{12}(c_2 + c_3\omega p + c_4\omega p^2 + c_5\omega^2 p);$$ 
$\omega$ (rad/s) is the engine speed disturbance; $\dot{m}_{af}$ is the inlet air mass flow rate, defined by:
$$\dot{m}_{af} = 2\hat{\theta}\sqrt{p/c_{10} - (p/c_{10})^2}; \text{ and}$$
$F_c$ is the commanded fuel input defined as $F_c = (1 + u(t))\dot{m}_c/\bar{\lambda}$, where $u(t)$ is the control input and $\bar{\lambda}$ is the ideal air/fuel ratio. 

Parameter values are: $c_1=0.41328$, $c_2=-0.366$, $c_3=0.08979$, $c_4=-0.0337$, $c_5=0.0001$, $c_6=2.821$, $c_7=-0.05231$, $c_8=0.10299$, $c_9=-0.00063$, $c_{10}=1$, $c_{12}=0.9$, $c_{25}=1$, $c_{26}=4$.


\section{Model of the Quadruple-Tank Process}\label{app:QT_model}
The dynamics of the Quadruple-Tank Process are given by the following set of ODEs \cite{QT00}:
\begin{align*}
\frac{dh_1}{dt} &= -\frac{a_1}{A_1}\sqrt{2g h_1} + \frac{a_3}{A_1}\sqrt{2g h_3} + \frac{\gamma_1k_1}{A_1}u_1\\
\frac{dh_2}{dt} &= -\frac{a_2}{A_2}\sqrt{2g h_2} + \frac{a_4}{A_2}\sqrt{2g h_4} + \frac{\gamma_2k_2}{A_2}u_2\\
\frac{dh_3}{dt} &= -\frac{a_3}{A_3}\sqrt{2g h_3} + \frac{(1-\gamma_2)k_2}{A_3}u_2\\
\frac{dh_4}{dt} &= -\frac{a_4}{A_4}\sqrt{2g h_4} + \frac{(1-\gamma_1)k_1}{A_4}u_1,
\end{align*}
where $h_i,a_i,A_i$, $i\in\{1,2,3,4\}$ are the water level, cross-section of the outlet hole, and cross-section of tank $i$, respectively. Inputs $u_1,u_2$ indicate the voltages applied to the pumps and the corresponding flows are $k_1u_1,k_2u_2$.
The parameters $\gamma_1,\gamma_2\in(0,1)$ show the settings of the valves. The flow to tank $1$ is $\gamma_1k_1u_1$ and the flow to tank $4$ is $(1-\gamma_1)k_1u_1$ (similarly for the other two tanks). The acceleration of gravity is denoted by $g$. The water levels of tanks $1,2$ are measured by sensors as $k_ch_1,k_ch_2$. The parameter values are: $A_1=A_3=28\,\text{cm}^2$, $A_2=A_4=32\,\text{cm}^2$, $a_1=a_3=0.071\,\text{cm}^2$, $a_2=a_4=0.057\,\text{cm}^2$, $k_c = 0.5\,\text{V}/\text{cm}$, $g = 9.81\,\text{m}/\text{s}^2$. We have chosen the steady state values $h_1^0 = 12.4 \,\text{cm} ,h_2^0 = 12.7\,\text{cm},h_3^0 = 1.8\,\text{cm},h_4^0 = 1.4\,\text{cm},u_1^0=3.00\,\text{V},u_2^0=3.00\,\text{V},k_1 = 3.33\,\text{cm}^3/\text{Vs}$, and $k_2 = 3.35\,\text{cm}^3/\text{Vs}$.
\end{document}